\documentclass[ejs,preprint]{imsart}
\usepackage{amsmath}
\usepackage{graphicx}
\usepackage[numbers]{natbib}
\usepackage{url}
\usepackage{amsthm,amssymb,float,amstext,bm,enumerate,xcolor,multirow}
\RequirePackage[colorlinks,citecolor=blue,urlcolor=blue]{hyperref}

\usepackage[normalem]{ulem}

\doi{10.1214/154957804100000000}
\pubyear{0000}
\volume{0}
\firstpage{1}
\lastpage{1}

\makeatletter
\newcommand*{\addFileDependency}[1]{
  \typeout{(#1)}
  \@addtofilelist{#1}
  \IfFileExists{#1}{}{\typeout{No file #1.}}
}
\makeatother

\startlocaldefs

\def\R{\mathbb{R}}
\DeclareMathOperator*{\argmin}{argmin}
\def\mc{\mathcal}
\def\mp{m_\oplus}
\def\hmp{\hat{m}_\oplus}
\def\son{\sum_{i = 1}^n}
\def\hmu{\hat{\mu}}
\def\hsigma{\hat{\sigma}}
\def\gp{g_\oplus}
\def\hgp{\hat{g}_\oplus}
\def\tgp{\tilde{g}_\oplus}
\def\Lp{\Lambda_\oplus}
\def\hLp{\hat{\Lambda}_\oplus}

\def\inv{^{-1}}
\def\d{\mathrm{d}}
\newcommand{\norm}[1]{\lVert #1 \rVert}

\theoremstyle{plain}
\newtheorem{theorem}{Theorem}
\newtheorem{lemma}{Lemma}

\newtheorem{cor}{Corollary}

\endlocaldefs

\begin{document}


\begin{frontmatter}

\title{Fréchet single index models for object response regression}
\runtitle{Fréchet Single Index Regression}

\author{\fnms{Aritra} \snm{Ghosal}\ead[label=e1]{ghosal@pstat.ucsb.edu}}
\and
\author{\fnms{Wendy} \snm{Meiring}\ead[label=e2]{meiring@pstat.ucsb.edu}}
\address{Department of Statistics and Applied Probability \\
University of California Santa Barbara \\
\printead{e1,e2}}

\author{\fnms{Alexander} \snm{Petersen}\thanksref{t1} \corref{}\ead[label=e3]{petersen@stat.byu.edu}}
\address{Department of Statistics \\
Brigham Young University \\
\printead{e3}}

\thankstext{t1}{A. Petersen was supported by NSF grant DMS 2128589}

\runauthor{Ghosal, Meiring and Petersen}

\begin{abstract}

With the increasing availability of non-Euclidean data objects, statisticians are faced with the task of developing appropriate statistical methods for their analysis. For regression models in which the predictors lie in $\R^p$ and the response variables are situated in a metric space, conditional Fréchet means can be used to define the Fréchet regression function.  Global and local Fréchet methods have recently been developed for modeling and estimating this regression function as extensions of multiple and local linear regression, respectively.  This paper expands on these methodologies by proposing the Fréchet Single Index model, in which the Fréchet regression function is assumed to depend only on a scalar projection of the multivariate predictor. Estimation is performed by combining local Fréchet along with M-estimation to estimate both the coefficient vector and the underlying regression function, and these estimators are shown to be consistent.  The method is illustrated by simulations for response objects on the surface of the unit sphere and through an analysis of human mortality data in which lifetable data are represented by distributions of age-of-death, viewed as elements of the Wasserstein space of distributions.

\end{abstract}



\tableofcontents

\end{frontmatter}

\section{Introduction}
\label{sec:intro}

A challenges in modern statistics is to analyze complex data objects that often possess structural or geometric properties.  Often, such properties are essential to their character and interpretation, and must be respected in statistical analyses to maintain maximal utility in drawing scientific conclusions.  A basic ingredient for modeling these objects is the presence of a metric that quantifies the disparity between them, from which one can extend valuable statistical concepts such as measures of center and dispersion.  These ideas date back to the seminal work of \cite{frechet1948elements}, where the Fréchet mean and variance were defined for random elements of a metric space.  In recent years, such data have been termed random objects \citep{muller2016peter}, while the associated set of tools has also been referred to as object oriented data analysis \citep{marron2014overview,patr:15}. Relevant examples include covariance matrices \citep{yuan:12}, probability distributions \citep{petersen2021wasserstein,chen2021wasserstein}, and networks \citep{mull:18:5}, among many others. 

The demand for regression tools for modeling the dependence of random objects on vector-valued covariates has grown steadily in recent years.  Along the way, the scope of these tools has broadened significantly, beginning with relatively simple spaces such as a circle or sphere \citep{fish:95,fish:87,chan:89}, then on to smooth manifolds \citep{pell:06,shi:09,niet:11,hink:12,yuan:12,flet:13,corn:17} and, most recently, general metric spaces \citep{fara:14,petersen2019frechet}.  In the case of smooth Riemannian manifolds, the cited regression models and estimators include parametric, semiparametric, and nonparametric variants that provide valuable flexibility; methods for general metric spaces are comparatively less developed.  \cite{petersen2019frechet} recently introduced two techniques applicable to response objects in a generic metric space.  Termed global and local Fréchet regression, these tools generalize linear and local linear regression, respectively, from the scalar response setting using similar principles to the classical Fréchet mean.  

As local and global Fréchet regression are extensions of classical tools for scalar response variables, one may naturally look to other scalar response models for inspiration in developing methods to balance the strengths and weaknesses of these two methods, the former providing flexibility, and the latter stability.  The model proposed in this paper is based on the single index model for scalar responses, specifically the approach studied by \cite{ichimura1993semiparametric}.  For a random pair $(X,Y) \in \R^p \times \R,$ the scalar response single index model asserts that
\begin{equation}
    \label{eq:SIM}
    m(x) := E(Y|X = x) = g(\theta_0'x)
\end{equation}
for an unknown coefficient $\theta_0 \in \R^p$ and unknown smooth function $g$. The hybrid nature is embodied by the finite-dimensional parameter $\theta_0$ and a univariate function $g$ that resides in an infinite-dimensional function space.  This is in contrast with multiple linear regression, in which $g$ is assumed to be linear, and a fully nonparametric model in which $m$ is a smooth function with $p$-dimensional domain.  Interpretation is simplified compared to a fully nonparametric model since $\theta_0$ is a global parameter that modulates the effect of each predictor. At the same time, the model is more flexible than a linear one by allowing the effect of the index $\theta_0'x$ to be nonlinear.  A theoretical advantage that adds utility to the model is that various reasonable estimators, including that studied in \cite{ichimura1993semiparametric}, are able to estimate $\theta_0$ with a parametric rate, even yielding a limiting normal distribution in some cases.  

In this paper, the Fréchet single index model is proposed as a generalization of the standard single index model through the use of conditional Fréchet means.  Prior to its definition, Section~\ref{sec:frechetReg} provides the necessary background on Fréchet means, both marginal and conditional, as well as a description of the local Fréchet regression technique for estimating the latter.  The FSI model is formally defined in Section~\ref{sec:FSI}, where estimators of the coefficient vector and univariate object-valued regression function are also detailed.  Consistency of both the single index parameter and the overall regression estimator are established.  Simulations for spherical response data illustrate the sampling variability of these estimators in Section~\ref{sec:sims}, and a real data analysis involving the association of distributions of age-at-death for various countries with economic indicators is provided in Section~\ref{sec:mort}.  Code for both the simulation and real data example can be found on Github (\url{https://github.com/aghosal89/Frechet_SingleIndex}).

\section{Background on Fréchet Regression}
\label{sec:frechetReg}

Let $(\Omega, d)$ be a bounded metric space.  The response $Y\in \Omega$ is to be modeled conditionally on a $p$-dimensional covariate $X\in \mathbb{R}^p$. Assume $(X,Y)\sim F$, with $F$ being a joint distribution on $\mathbb{R}^p \times \Omega$ such that $\Sigma=\operatorname{Var}(X)$ exists with $\Sigma$ positive definite and $\mu=E(X)$. When $\Omega$ is a Euclidean space such as $\mathbb{R}^d$ or $L^2[0,1]$ as would be the typical case for multivariate or functional data, one can utilize the usual notions of expectation arising from Lebesgue integration to quantify the mean and variance of $Y$.  For arbitrary metric spaces $\Omega,$ the concepts of mean and variance of a random variable are replaced by the Fréchet mean and variance \citep{frechet1948elements}, respectively, defined as
\begin{equation}
\label{eq:fr_meanvar}
\omega_{\oplus}=\underset{\omega \in \Omega}{\operatorname{argmin}} \, E\left(d^{2}(Y, \omega)\right), \quad V_{\oplus}=E\left(d^{2}\left(Y, \omega_{\oplus}\right)\right).
\end{equation}
Existence and uniqueness of the Fréchet mean is by no means guaranteed for general metric spaces.  However, in special cases such as certain Riemannian manifolds \citep{afsa:11,pennec2018barycentric} or spaces with negative curvature \citep{bhat:03,bhat:05}, Fréchet means exist and are unique.  For the moment, we assume at least that a minimizer exists, with the consequence that $\omega_\oplus$ and $V_\oplus$ are not vacuous, and the latter is unique.  Extending these concepts to regression, define the Fréchet regression function $Y$ given $X=x \in \mathbb{R}^p$ as
\begin{equation}
    m_{\oplus}(x)=\underset{\omega \in \Omega}{\argmin}\, M_{\oplus}(\omega,x), \quad M_{\oplus}(\cdot,x)=E(d^2(Y,\cdot)|X=x).
    \label{eq:fr_regressionfunction}
\end{equation}

\subsection{Local Fréchet Regression}
\label{ss:LF_def}

Two different approaches were proposed by \cite{petersen2019frechet} to estimate the conditional Fréchet means $m_\oplus(x).$  First, a global model for $m_\oplus(x)$ was proposed in which $\mp(x)$ can be written as the minimizer of an alternative objective function motivated by multiple linear regression in the case $\Omega = \R.$  The result is that $\mp(x)$ can be viewed as a weighted Fréchet mean, where the weights depend on the joint distribution $F$ and the input $x.$  As a direct generalization of linear regression, global Fréchet regression similarly can be overly restrictive for random object responses.  Thus, in a second approach, \cite{petersen2019frechet} also demonstrated how to generalize local linear regression to estimate $m_\oplus(x)$ under less restrictive assumptions on the function $\mp$.  This approach, termed local Fréchet regression, will now be described.

The motivation stems from considering a scalar predictor $X \in \R$ and response $Y \in \Omega = \mathbb{R},$ so that the target $m_{\oplus}(x)=: m(x)$ in (\ref{eq:fr_regressionfunction}) is just the usual conditional expectation. Let $K$ be a probability density kernel, $h$ a bandwidth, and $K_h(\cdot) = h^{-1}K(\cdot/h)$, as used in local polynomial estimation.  Given a random sample $(X_i,Y_i),$ $i = 1,\ldots,n$ and a fixed predictor value $x$, \cite{petersen2019frechet} utilized the expression
$$
\hat{l}(x) = \frac{1}{n}\sum_{i = 1}^n \hat{s}_h(X_i, x)Y_i
$$
for the well-known local linear estimator \citep{fan1996local} of $\mp(x)$ in order to motivate the local Fréchet technique for general response object spaces $\Omega.$  Here, the empirical weight function $\hat{s}_h$, as derived from the local linear least squares criterion, is
\begin{equation}
    \label{eq:lin_local_weights}
    \hat{s}_h(z, x) = \hat{\varsigma}^{-2}K_h(z - x)\left[\hmu_2 - \hmu_1(z - x)\right],
\end{equation}
where
$$
\hmu_j = n\inv\son K_h(X_i - x)(X_i - x)^j, \quad \hat{\varsigma}^2 = \hmu_0\hmu_2 - \hmu_1^2,
$$
and thus satisfies $n\inv \sum_{i = 1}^n \hat{s}_h(X_i, x) = 1.$ Hence, $\hat{l}(x)$ is a weighted average of the observed responses or, equivalently, 
\begin{equation}
    \label{eq:lin_local}
    \hat{l}(x) = \argmin_{y \in \R} \son \hat{s}_h(X_i, x)(Y_i - y)^2.
\end{equation}

The local Fréchet regression estimator of $\mp(x)$ in \eqref{eq:fr_regressionfunction} for a general metric space $\Omega$ is obtained by replacing the squared difference $(Y - y)^2$ in \eqref{eq:lin_local} by its appropriate counterpart in metric spaces, the squared distance.  Given a random sample $(X_1,Y_1),(X_2,Y_2),\ldots,(X_n,Y_n)$ independently distributed according to $F$ and a fixed $x \in \mathbb{R}$, the local Fréchet estimator is
\begin{equation}
\hat{l}_\oplus(x)= \argmin_{\omega \in \Omega} \sum_{i=1}^{n} \hat{s}_h(X_i, x) d^{2}\left(Y_{i}, \omega\right)
\label{eq:fr_local_empirical}
\end{equation}
where the weights are again given by \eqref{eq:lin_local_weights}. As pointed out by one reviewer, the criterion minimized in the right-hand side of \eqref{eq:fr_local_empirical} is, for each $x$ and $\omega$, a local linear estimator of the conditional expected value represented by $M_\oplus(\omega,x)$ in \eqref{eq:fr_regressionfunction}.  Thus, the local Fréchet regression approach is equivalent to pointwise estimation of $M_\oplus$ by local linear regression, followed by its minimization over $\Omega.$

\section{The Fréchet Single Index Model}
\label{sec:FSI}

While extension of the local Fréchet estimator to accommodate multivariate predictors $x \in \R^p,$ $p > 1,$ is mathematically straightforward, its performance will deteriorate quickly with increasing $p$ due to the curse of dimensionality.  Thus, for even moderate $p$, the global Fréchet model may be preferable despite its bias due to increased stability in estimation.  Unsurprisingly, one can attempt to balance the strengths, and mitigate the weaknesses, of these two Fréchet approaches in the same spirit that semiparametric techniques do so for parametric and nonparametric estimators in classical models.  Specifically, the proposed Fréchet Single Index model assumes that the Fréchet regression function only depends on $x$ through an index $\theta_0'x \in \R,$ for some $\theta_0 \in \R^p.$

\subsection{Model Definition}
\label{ss:FSI_model}

The coefficient $\theta_0 \in \mathbb{R}^p$ constitutes the primary target of interest in this new model, as it lends interpretability by specifying the contribution of each predictor.  For identifiability purposes \citep{lin2007identifiability}, define the parameter space $$\Theta_p=\{\theta\in \mathbb{R}^p: \text{the first non-zero element of $\theta$ is positive, and } \|\theta\|_E=1 \},$$ where $\|\cdot\|_E$ is the Euclidean norm. Hence, $\theta$ belongs to the surface of the unit sphere in $p$ dimensions.  By this convention, $\Theta_1 = \{1\}$, for which the required theoretical work is well-developed as local Fréchet regression. Therefore we focus on analyzing $p \ge 2$. 

A comprehensive discussion of a large class of single index models and their applications can be found in \cite{ichimura1993semiparametric} for response data $Y \in \mathbb{R}$, where the index parameter $\theta_0$ was estimated using the Semiparametric Least Squares (SLS) method. The procedure that will be described for estimating the coefficient in the proposed model is inspired by this intuitive technique, and leverages local Fréchet regression and standard distance-based least squares.

To formally define the new model, let $F_X$ denote the marginal distribution of $X$, with support $\mc{X} \subset \R^p$. For any $\theta \in \Theta_p,$ define the Fréchet regression function conditional on the projected variable $\theta'X$ as 
\begin{equation}
g_{\oplus}(u, \theta)= \argmin_{\omega \in \Omega} \Lambda_{\oplus}(\omega,u,\theta),\quad \Lambda_\oplus( \cdot, u,\theta) = E(d^2(Y, \cdot)|\theta'X = u),  
\label{eq:gplus}
\end{equation}
where $u \in \mc{U}_\theta := \{\theta'x: \, x \in \mc{X}\}$ and a minimizer is assumed to exist.  Thus, the Fréchet single index (FSI) model for $m_\oplus(x)$ in \eqref{eq:fr_regressionfunction} is 
\begin{equation}
    \label{eq:fsi_model}
    m_\oplus(x) = g_\oplus(\theta_0'x,\theta_0),
\end{equation}
consisting of an unknown parameter $\theta_0 \in \Theta_p$ and unknown smooth function $g_\oplus(\cdot,\theta_0)$ on $\mc{U}_{\theta_0}$. 

Given existence of the minimizers in \eqref{eq:gplus}, identifiability of the parameter $\theta_0$ is equivalent to the statement
$$
P(\gp(\theta'X, \theta) \neq \gp(\theta_0'X, \theta_0)) > 0,
$$
from which it can be deduced that
\begin{equation}
    \label{eq:W}
    W(\theta) = E\left(d^2(Y, g_\oplus(\theta'X,\theta))\right),
\end{equation}
the natural generalization of the least-squares criterion for metric spaces, is uniquely minimized at $\theta_0.$  Thus, the above criterion will be used to construct an M-estimator for $\theta_0.$  In a recent preprint, \cite{bhattacharjee2021single} independently investigated model \eqref{eq:fsi_model}, though using a slightly different strategy to estimate $W(\theta)$ than that employed in this paper. A comparison of the proposed estimator and that of \cite{bhattacharjee2021single} is provided below in Section~\ref{ss:Comp}.

\subsection{Estimation}
\label{ss:FSI_est}

Suppose a random sample $(X_i, Y_i),$ $i = 1,\ldots, n$, distributed according to $F$ is available. As the true parameter $\theta_0$ is unknown, we proceed to estimate the target in \eqref{eq:fsi_model} in two steps.  First, $\gp(\theta'x,\theta)$ is estimated for fixed $\theta$ using local Fréchet regression, followed by optimization over $\theta$.  Let $h > 0$ be a given bandwidth and $K$ a univariate probability density kernel, as before. The estimates in this section depend on $h$, although we suppress this dependence for simplicity in several formulae. 

For a fixed $\theta \in \Theta_p,$ repurposing \eqref{eq:lin_local_weights} and \eqref{eq:fr_local_empirical} for use with the predictors $\theta'X_i,$ we obtain the estimate
\begin{equation}
    \label{eq:gplus_est}
    \hgp(\theta'x,\theta)= \argmin_{\omega \in \Omega} \hLp (\omega,\theta'x,\theta), \quad \hLp(\omega,\theta'x,\theta)= \frac{1}{n}\sum_{i=1}^{n} \hat{r}_h(X_i, x, \theta) d^{2}\left(Y_{i}, \omega\right).
\end{equation}
Here, the weight function $\hat{r}_h: \R^p \times \R^p \times \Theta_p \rightarrow \R$ is
\begin{equation}
    \label{eq:fsi_weights}
    \hat{r}_h(z, x, \theta) = \hsigma_{\theta}^{-2}(x)K_h(\theta'(z - x))\left[\hmu_{2,\theta}(x) - \hmu_{1,\theta}(x)(\theta'(z - x))\right],
\end{equation}
where, for $j = 0, 1, 2,$ 
\begin{equation}
    \label{eq:muj_est}
    \hmu_{j,\theta}(x) = n\inv\son K_h(\theta'(X_i - x))(\theta'(X_i - x))^j
\end{equation} 
and $\hsigma_{\theta}^2(x) = \hmu_{0,\theta}(x)\hmu_{2,\theta}(x) - \hmu_{1,\theta}(x)^2.$

Utilizing this result, we construct a criterion for estimating $\theta_0$ by defining an empirical version of \eqref{eq:W}. Replacing the expectation with the empirical distribution, and replacing $\gp(\theta'X_i, \theta)$ with the fitted value $\hat{Y}_i(\theta,h) = \hgp(\theta'X_i, \theta)$ yields
\begin{equation}
    \label{eq:What}
    W_n(\theta)=\frac{1}{n}\sum_{i=1}^{n}d^2(Y_i, \hat{Y}_{i}(\theta,h)).
\end{equation} 
The coefficient vector $\theta_0$ is then estimated by
\begin{equation}
    \label{eq:theta_est}
    \hat{\theta} = \hat{\theta}(h) = \argmin_{\theta \in \Theta_p} W_n(\theta).
\end{equation}

As is typically the case in this type of semi-parametric estimation approach, the bandwidth $h$ cannot decay too quickly if one is to obtain a consistent estimator of $\theta_0$.  Indeed, Theorem~\ref{thm:unif_cons} below restricts the decay of $h$ in a way that depends on the dimension $p$ as well as $n$. Nevertheless, in constructing the final estimator $\hmp(x)$ of the regression function $\mp(x),$ a different smoothing bandwidth may be used, potentially improving the overall rate of convergence.  Specifically, denote by $\tilde{g}(\theta'x,\theta)$ the estimator in \eqref{eq:gplus_est} for any $\theta$ and $x$ using a bandwidth $\tilde{h} > 0.$  Then the final regression estimator is 
\begin{equation}
    \label{eq:FSI_final}
    \hmp(x) = \tgp(\hat{\theta}'x,\hat{\theta}).
\end{equation}

\subsection{Theoretical Properties}
\label{ss:theory}

For semiparametric models such as the proposed FSI model, the primary target of interest is the parametric component, in this case $\theta_0.$  Once the properties of the estimate $\hat{\theta}$ are known, their effects on the ensuing estimate $\hmp(x)$ in \eqref{eq:FSI_final} can be determined. A necessary preliminary result is the uniform consistency of the estimates $\hgp(\theta'x,\theta)$ in \eqref{eq:gplus_est} over $x$ and $\theta,$ in analogy to Theorem 5.1 of \cite{ichimura1993semiparametric} in the case of a scalar response.  Uniform consistency of local Fréchet regression for a scalar predictor was first proved by \cite{chen2022uniform}, combining pointwise results of \cite{petersen2019frechet} with uniform results on kernel estimation \citep{silverman1978weak,mack1982weak}.  For the proposed estimators in the FSI model, this result of \cite{chen2022uniform} implies consistency of $\hgp(\theta'x,\theta)$ that is uniform in $x$ for a fixed $\theta$, and is thus insufficient for our purposes.  

The key results in \cite{silverman1978weak,mack1982weak} utilize Brownian bridge approximations to the empirical distribution of a scalar or bivariate sample \citep{komlos1975approximation,tusnady1977remark}.  The quantities $h^{-j}\hmu_{j,\theta}(x)$ are, in fact, kernel estimators of the density of $\theta'X$ at $\theta'x,$ so the appropriate generalization of the Brownian bridge technique involves such an approximation that is uniform in $x$ and $\theta.$  We obtain such an approximation by utilizing a comparatively lesser known multivariate Brownian bridge approximation of \cite{csorgHo1975new} for the predictor sample $X_1,\ldots,X_n$.  The value of this particular approximation was demonstrated in \cite{rosenblatt1976maximal} for $p$-dimensional kernel density estimates.  However, to the knowledge of the authors, the corresponding result of Lemma~\ref{lma:BB_approx}, stated and proved in Section~\ref{ss:aux}, is a novel result of independent interest.  It establishes a strong Brownian bridge approximation for kernel density estimates of projected variables $\theta'X$ that is uniform in both the projection direction and the density argument, and is crucial for establishing Theorem~\ref{thm:unif_cons} below.  As we are performing regression, Lemma~\ref{lma:BB_approx2} in Section~\ref{ss:aux} gives a similar approximation in the spirit of Proposition~4 of \cite{mack1982weak} for the scalar response case.  We require the following conditions.
\begin{itemize}
    \item[(K)] The kernel $K$ is a probability density function  with $K(w) = K(-w)$, is uniformly continuous, and is of bounded variation.  With $K^{(j)}(w) = K(w)w^j$, $j = 0,1,2,$ the definite integrals $\int_\R w^4K(w)\mathrm{d}w,$ $\int_\R w^6 K^2(w)\mathrm{d}w,$ and $\int_\R |w\log|w||^{1/2}\mathrm{d}K^{(j)}(w)$ are all finite.
    \item[(F)] Let $F_X$, $F_{\theta'X}$, and $F_{\theta'X|Y}$ denote the distributions of $X,$ $\theta'X,$ and $\theta'X|Y$, respectively. 
    \begin{enumerate}[i)]
        \item The support $\mc{X}$ is a bounded set, and the Rosenblatt transformation $T:\mc{X} \rightarrow \mathbb [0,1]^p$ defined in \cite{rosenblatt1952remarks} as $T_1(x) = F_{X_1}(x_1),$ $T_j(x) = F_{X_j|X_1,\ldots,X_{j-1}}(x_j|x_1,\ldots,x_{j-1})$ for $j =2,\ldots,p,$ is Lipschitz continuous.  The space $\Omega$ is also bounded with respect to the metric $d$.
        \item For each $\theta,$ the support of $F_{\theta'X}$ is a compact interval $\mathcal{U}_\theta,$ and $F_{\theta'X}$ admits a density $f_{\theta'X}$ that is twice differentiable on the interior of $\mathcal{U}_\theta$ and satisfies 
        $$
        \inf_{\theta \in \Theta_p} \inf_{u \in \mathcal{U}_\theta} f_{\theta'X}(u) > 0, \quad \sup_{\theta \in \Theta_p} \sup_{u \in \mathcal{U}_\theta^\mathrm{o}} |f_{\theta'X}''(u)| < \infty.
        $$
        \item For any $y \in \Omega$ and $\theta \in \Theta_p,$ the conditional density $f_{\theta'X|Y}(u|y) = (\partial/\partial u) F_{\theta'X|Y}(u|y)$ is twice differentiable with respect to $u$ on $\mc{U}_\theta^\mathrm{o}$ and satisfies $$\sup_{y \in \Omega} \sup_{\theta \in \Theta_p} \sup_{u \in \mc{U}_\theta^{\mathrm{o}}} \left| \frac{\partial^2}{\partial u^2} f_{\theta'X|Y}(u|y)\right| < \infty.$$  Additionally, for any open set $V \subset \Omega$ and $\theta \in \Theta_p,$ the function \mbox{$P(Y \in V | \theta'X = \cdot)$} is continuous on $\mc{U}_\theta^\mathrm{o}.$
        \item For each fixed $\omega \in \Omega,$ with $R = d^2(Y,\omega)$, the vector $(X, R)$ has density $f_{X,R}(x,r)$ such that $\sup_x \int_\mathbb{R} r^2 f_{X,R}(x,r)\d r < \infty$.  Furthermore, the Rosenblatt transformation $T^+$ of $(X,R)$ is Lipschitz continuous.
    \end{enumerate}
    \item[(M)] The minimizers in \eqref{eq:gplus} are unique and, for any $\epsilon > 0,$ there is an $\eta = \eta(\epsilon) > 0$ such that
    $$
    \inf_{\theta \in \Theta_p} \inf_x \inf_{d(\omega, \gp(\theta'x,\theta)) > \epsilon} \{\Lp(\omega, \theta'x, \theta) - \Lp(\gp(\theta'x,\theta), \theta'x,\theta)\} > \eta.
    $$
    Finally, all minimizers \eqref{eq:gplus_est} exist with probability approaching 1, though these need not be unique.
\end{itemize}
Assumption (K) is common for smoothing estimators and strengthens the corresponding assumption of \cite{petersen2019frechet} for local Fréchet regression in order to provide uniform consistency.  Assumption (F) lists distributional assumptions on $(X,Y)$.  Parts i) and iv) are essential in order to leverage the Brownian bridge approximations of Lemmas~\ref{lma:BB_approx} and \ref{lma:BB_approx2}, respectively; parts ii) and iii) control the behavior of the estimated weight functions $\hat{r}_h$ in \eqref{eq:fsi_weights}, and imply the consistency of the empirical criteria in \eqref{eq:gplus_est}.  

Assumption (M) is a generic condition for M-estimators like those used here for the FSI model, though sufficient conditions for specific spaces $\Omega$ and distributions $F$ need to be derived on a case-by-case basis.  Compared to the corresponding assumptions employed by \cite{chen2022uniform} to establish uniform consistency of local Fréchet estimates, (M) is stronger in its uniformity over $\theta,$ but weaker as it only makes the separation assumption, given in the first display of (M), for the population criterion $\Lp$ and not the empirical criteria $\hLp.$  The verification of assumption (M) can be challenging and, in general, depends on properties of the metric space as well as the probability measure.  However, in some cases, uniqueness of both population and sample Fréchet means can be established, along with the separation assumption displayed in (M).  As a primary example, distributions on non-positively curved spaces, also known as Hadamard spaces, are known to possess unique Fréchet means \citep{sturm2003probability}.  These include Hilbert spaces of finite or infinite dimension, as well as convex subsets of these, as is the case for the space of Wasserstein distributions on the real line utilized below in the data example of Section~\ref{sec:mort}. Another example is the space of phylogenetic trees \citep{billera2001geometry}.  However, many data examples lie in spaces with positive curvature, such as spherical data that are illustrated in the simulations of Section~\ref{sec:sims}. While not a Hadamard space, the sphere is a proper Alexandrov space, for which sufficient conditions for uniqueness of Fréchet means, as well as the separation property stated in (M), have been investigated \citep{ohta2012barycenters}.

Observe also that (M) only requires uniqueness of the population Fréchet means, whereas only existence is required for the minimizers in \eqref{eq:gplus_est}.  Indeed, if closed balls in $\Omega$ are compact, the existence of minimizers $\hgp(u,\theta)$ follows by a continuity argument, so that the last statement of (M) ceases to be an assumption in this case.  The same is true if $\Omega$ is a closed, convex subset of a Hilbert space, with $d$ the Hilbertian metric.  In case multiple minimizers in \eqref{eq:gplus_est} exist, the consistency result below will hold for any such (sequence of) minimizers.  The proof of this result, given in Section~\ref{ss:main_proofs}, is considerably simplified compared to the uniform consistency arguments in both \cite{petersen2019frechet} and \cite{chen2022uniform} for global and local Fréchet regression, respectively.

\begin{theorem}
    \label{thm:unif_cons}
    Suppose assumptions (K), (F), and (M) are satisfied, and that $h \rightarrow 0$ as $n \rightarrow \infty$ such that $-\log(h)/(nh)$ and $[\log(n)]^3/[h^2n^{(p + 3)/(p+2)}]$ converge to 0 as $n \rightarrow \infty.$  Let $\hgp(\theta'x,\theta)$ denote any minimizer of \eqref{eq:gplus_est} when such a minimizer exists.  Then $$\sup_{\theta \in \Theta_p}\sup_{x} d(\hgp(\theta'x,\theta), \gp(\theta'x,\theta)) = o_P(1).$$
\end{theorem}

Theorem~\ref{thm:unif_cons} imposes two requirements on the $h$ besides the usual condition $h \rightarrow 0$ that ensures the smoothing bias goes to zero.  The first is that $-\log(h)/(nh) \rightarrow 0,$ and arises from the continuity modulus of a Brownian bridge in dimension $p + 1$.  This is the same rate that arises in one-dimensional smoothing, as the modulus of continuity is insensitive to the underlying dimension.  On the other hand, the approximation error leads to the condition $[\log(n)]^3/[h^2n^{(p+3)/(p+2)}] \rightarrow 0$ that does indeed depend on the dimension of the predictor variable.  Although such a condition is not necessary to derive consistency for scalar responses, i.e. $\Omega = \R,$ the lack of a closed form expression for the minimizers in \eqref{eq:gplus_est} prohibits arguments available in this special case from being used for general object response spaces.  We also remark that, in the case $p = 2,$ this dimension-dependent requirement on the bandwidth can be weakened using the specialized approximation result of \cite{tusnady1977remark} for bivariate distributions rather than the more general result of \cite{csorgHo1975new}.   

Having established uniform consistency of the local Fréchet regression estimates, one can easily demonstrate that the coefficient estimate is consistent, as well as the overall regression estimator $\hat{m}_\oplus(x)$ in \eqref{eq:FSI_final}. The proofs are given in Section~\ref{ss:main_proofs}.
\begin{cor}
    \label{cor:theta_cons}
    Suppose the assumptions of Theorem~\ref{thm:unif_cons} hold and that $\theta_0$ is identifiable. Then $\hat{\theta}$ converges to $\theta_0$ in probability.
\end{cor}
\begin{cor}
    \label{cor:reg_cons}
    Suppose the assumptions of Theorem~\ref{thm:unif_cons} hold.  Let $\tilde{h}$ be the bandwidth used to construct $\hmp(x)$ in \eqref{eq:FSI_final} and let $\tgp(\theta'x,\theta)$ denote the estimator in \eqref{eq:gplus_est} computed using this bandwidth.  If there exists $\delta > 0$ such that
    $$
    \sup_{\norm{\theta - \theta_0} < \delta} \sup_x d(\tgp(\theta'x,\theta), g_\oplus(\theta'x,\theta)) = o_P(1),
    $$
    then $\sup_x d(\hmp(x),\mp(x)) = o_P(1).
    $
\end{cor}
Regarding the condition on the bandwidth $\tilde{h}$ in the Corollary~\ref{cor:reg_cons}, observe that it will immediately hold for any $\delta > 0$ if $\tilde{h} = h.$ However, it may be possible for $\tilde{h}$ to decay more quickly than $h$ since uniform convergence is only required for $\theta$ near $\theta_0$ and not for all $\theta.$  A more precise specification of the potential gains requires further analysis.  As a starting point, one must obtain a rate of convergence for $\hat{\theta}.$  The usual approaches of either expanding $(W_n - W)(\theta)$ or controlling its continuity modulus near $\theta_0$ present non-trivial challenges due to the presence of estimates $\hat{Y}_i(\theta,h)$ in $W_n$, as these approaches would require some level of smoothness of these estimates in $\theta.$  It is conceivable that such properties could be derived efficiently for certain classes of smooth spaces, such as Riemannian manifolds, but we do not pursue these here.

\subsection{Qualitative Comparison with an Alternative Estimator}
\label{ss:Comp}

As mentioned previously, model \eqref{eq:fsi_model} is also studied independently in a recent preprint \citep{bhattacharjee2021single}.  The key difference in their approach is in the estimation of the objective function $W$ in \eqref{eq:W}.  Rather than averaging the prediction errors for each $\theta$ across all observations as in \eqref{eq:What}, \cite{bhattacharjee2021single} propose to bin the projected covariates $\theta'X_i$ into $M$ bins, where $M$ grows slowly in comparison to $n$.  Each bin is represented by a single pair of values $(X_m^*, Y_m^*)$, with $X_m^*$ being empirical mean of the predictors and $Y_m^*$ the emprical Fréchet mean of the responses in the $m$-th bin, respectively.  Then, \eqref{eq:What} is replaced by a similar version that averages prediction errors across the $M$ representative rather than the $n$ observed data points.

From a practical perspective, the choice to bin the data comes with some complications and additional choices that need to be made by the analyst, not least being the number of bins and placement of breaks, which can be difficult for large data sets, especially considering the different distributions of data points that can occur when varying the projection direction $\theta.$  While the referenced preprint does not give much motivation for the choice to bin, doing so could have some advantages in establishing theoretical properties.  For instance, in attempting to establish a central limit theorem for $\sqrt{n}(W_n(\theta) - W(\theta))$, one is again faced with the difficulty that $W_n$ involves the intermediate estimates $\hat{Y}_i(\theta,h).$  Unlike local linear estimators for Euclidean responses, these local Fréchet estimates have no closed form expression for general metric spaces.  Moreover, one cannot exploit properties of squared distances in Euclidean spaces to control the differences $d^2(Y_i, \hat{Y}_i(\theta,h)) - d^2(Y_i, g_\oplus(\theta'X_i, \theta))$ in the usual way.  However, if $\Omega$ is bounded, the reverse triangle inequality does yield 
$$
|d^2 (Y_i, \hat{Y}_i(\theta,h)) - d^2(Y_i, g_\oplus(\theta'X_i, \theta))| \leq 2\mathrm{diam}(\Omega) d(\hat{Y}_i(\theta,h), g_\oplus(\theta'X_i, \theta)),
$$
which can be controlled uniformly in $\theta$ and $X_i$ according to Theorem~\ref{thm:unif_cons}, but unfortunately not at a rate that is negligible compared to $\sqrt{n}.$  By binning the data, the same approach may yield an asymptotic limit as long as the uniform rate of the local Fréchet estimates shrinks faster than the effect of the increase in number of bins.  Nevertheless, it is unclear whether this represents a real phenomenon or is merely an artifact of analytic approach.

\section{Simulation Study on Spherical Data}
\label{sec:sims}

We implement our methodology when the responses lie on a Riemannian manifold object space. Let $\Omega=S^2$, the surface of the unit sphere in $\mathcal{R}^3$, with origin being the center. For any two points $y_1,y_2 \in S^2$, the geodesic distance between them is $d(y_1,y_2)=\arccos{(y_1'y_2)}$. We refer to a simulation setting as a unique combination of the sample size $n$, covariate dimension $p$, and noise level $\sigma^2 > 0$ that will be defined below.

\subsection{Data Generation}
For a given setting $(n, p, \sigma^2)$, independent and identically distributed data pairs $(X_i, Y_i) \in \mathbb{R}^p\times S^2$, $i = 1,\ldots,n$ were generated according to the following steps.

\begin{enumerate}
\item Independently generate predictor components $X_{ij},$ $j = 1,\ldots,p,$ as $X_{ij} = W_{ij}/\sqrt{p},$ where $W_{ij} \overset{\mathrm{iid}}{\sim} \mathcal{U}(-1,1)$.

\item With $\theta_0=(\theta_{01},\theta_{02},..., \theta_{0p})'$ being the true parameter, compute the latent predictor $U_i=\theta_{0}'X_i.$

\item Compute the conditional Fréchet mean at $X_i$, depending only on $U_i$, as
$$
\mp(X_i) = \left(\sqrt{\left(1-\frac{U_i^2}{{p}}\right)} \cos \left( \frac{\pi U_i}{\sqrt{p}}\right),\sqrt{\left(1-\frac{U_i^{2}}{{p}}\right)} \sin \left(\frac{\pi U_i}{\sqrt{p}}\right), \frac{U_i}{\sqrt{p}}\right).
$$

\item Generate a noise vector $Z_i$ as follows.  First, let $(V_{i1}, V_{i2})$ be an orthonormal basis for the tangent space $\mathrm{span}\{m_\oplus(X_i)\}^{\perp}$. Next, for a given noise level $\sigma^2,$ generate $C_i=(c_{i1},c_{i2})' \stackrel{iid}{\sim} N_2(\boldsymbol{0},\sigma^2 \boldsymbol{I}_2)$. Finally, set $Z_i=c_{i1}V_{i1}+c_{i2} V_{i2}$.

\item Generate the spherical response variable as
$$
Y_i =\cos \left(\left\|Z_{i}\right\|_{E}\right) m_\oplus(X_i) + \sin \left(\left\|Z_{i}\right\|_{E}\right) \frac{Z_{i}}{\left\|Z_{i}\right\|_{E}}.
$$
\end{enumerate}

\noindent
Steps 4 and 5 produce a point $Y_i$ on the sphere with conditional Fréchet mean equal to $m_\oplus(X_i).$  To give an idea of what the responses look like relative to the conditional Fréchet mean function for a given noise level, Figure~\ref{fg4.1} shows example data sets and corresponding estimates for $p = 5$ under two noise scenarios ($\sigma^2=0.4$ and $\sigma^2=0.8$) and three sample sizes ($n = 50, 100, 200$).
\newline
\begin{figure}[p]
\centering
\includegraphics[width=4.8cm]{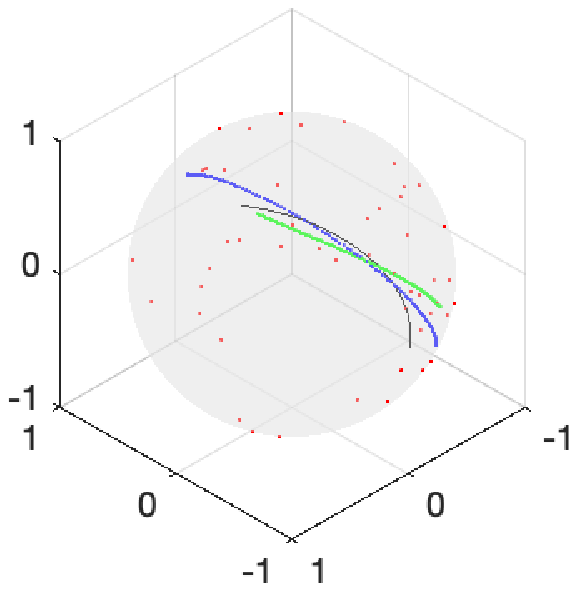}
\includegraphics[width=4.8cm]{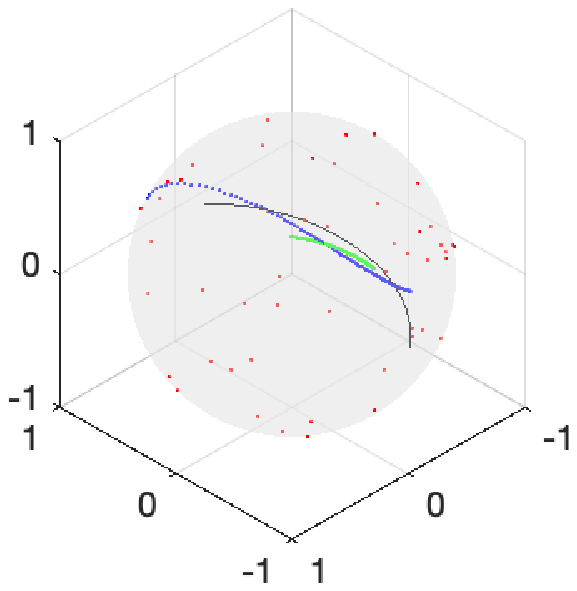} \\
\includegraphics[width=4.8cm]{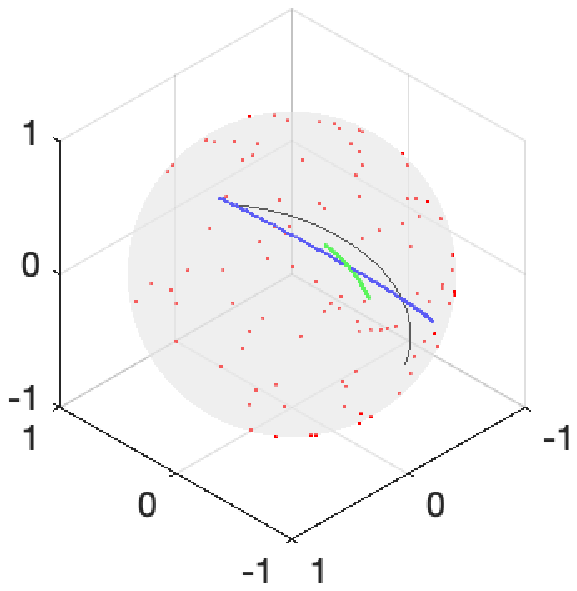}
\includegraphics[width=4.8cm]{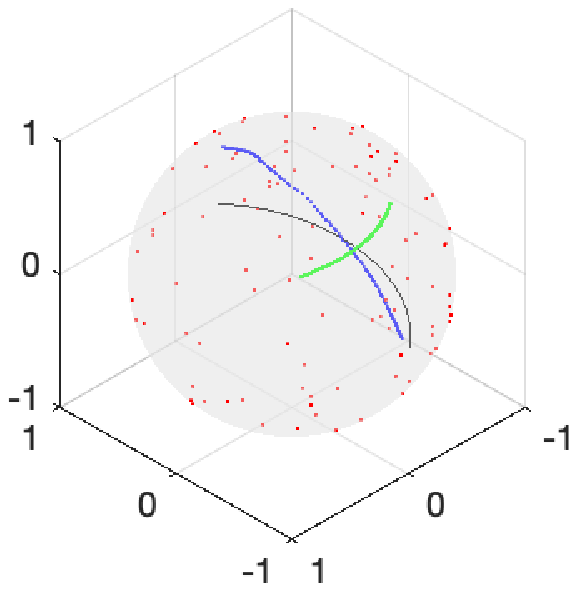} \\
\includegraphics[width=4.8cm]{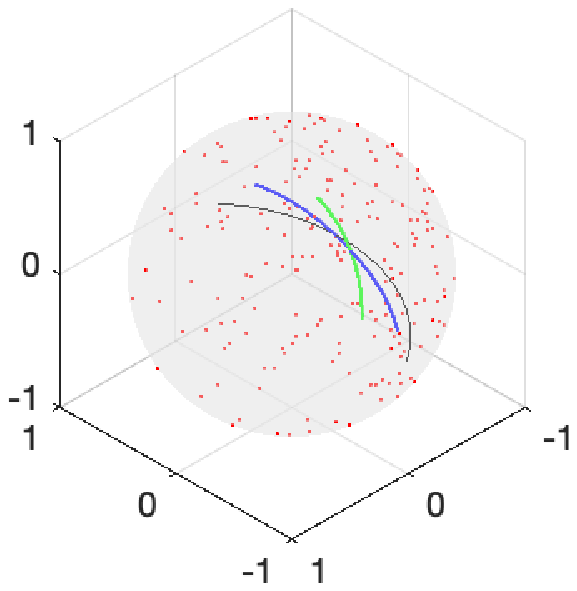}
\includegraphics[width=4.8cm]{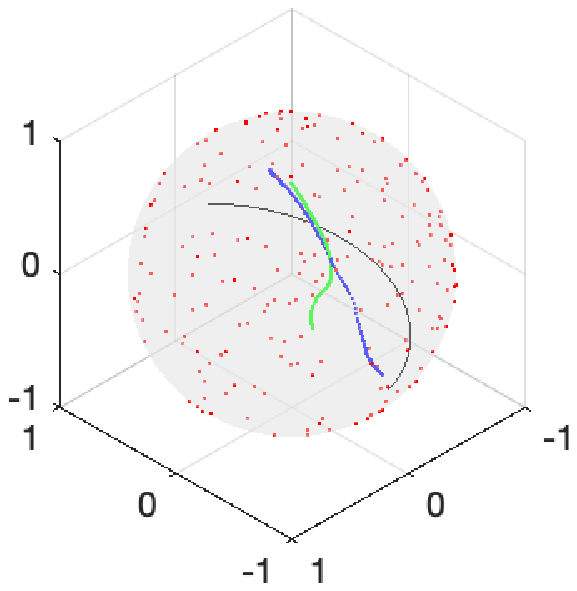}
\caption{\footnotesize Examples of simulated data sets for covariate dimension $p=5$, corresponding to sample sizes $n=50$ (top row), $n = 100$ (middle row), and $n=200$ (bottom row), and noise levels $\sigma^2 = 0.4$ (left column) and $\sigma^2 = 0.8$ (right column). The red dots represent values of $Y_i$ in the sample, while the regression function values $m_\oplus(x)$ are shown by the black curve for $x \in [0,1]^p$. The blue dots represent the FIS fitted responses for $n$ observations using \eqref{eq:FSI_final}. The green dots are the fitted responses obtained by computing \eqref{eq:gplus_est} for a value of $\theta$ far from the true value $\theta_0.$ }
\label{fg4.1}
\end{figure}

\subsection{Computational Details}
For each simulated data set, estimation was performed using a grid for the bandwidth $h$. For given values $\theta$ and $h$, the local Fréchet estimate $\hat{g}_{\oplus}(u,\theta)$ in \eqref{eq:gplus_est} was obtained for values $u =\theta'X_i,$ $i = 1,\ldots,n,$ using a non-convex optimization trust region algorithm as implemented in ManOpt toolbox for Matlab \citep{boumal2014manopt,petersen2019frechet}.  As the algorithm requires an initial estimate, we computed the leave-one-out Nadaraya-Watson estimate
$$
\tilde{{Y}}_{(i)}^{(NW)}\left(h, \theta\right) = \frac{\sum_{l \neq i} Y_{l} K\left(\left[X_{i}' \theta-X_{l}'\theta \right] / h\right)}{\sum_{l \neq i}  K\left(\left[X_{i}' \theta-X_{l}'\theta\right] / h\right)}
$$
for each observed predictor values $X_i.$  Then, the initial estimate that is entered into the algorithm is obtained by projecting onto the sphere, i.e.
$$
\hat{Y}_{(i)}^{(0)}(h,\theta)=\frac{\tilde{{Y}}_{(i)}^{(NW)}\left(h, \theta\right)}{\|\tilde{{Y}}_{(i)}^{(NW)}\left(h, \theta\right)\|_{E}}.
$$

Computation of the estimate $\hat{\theta}(h)$ by optimizing the criterion $W_n$ is the more challenging task, particularly for larger values of $p$, since there is no explicit form for the gradient or Hessian.  Numeric evaluation of the gradient can also be quite expensive when $n$ is large due to the need to repeatedly perform local Fréchet regression for each data point.  In addition, any optimization procedure is sensitive to the starting value for $\theta,$ particularly for larger $p$, further increasing the computational burden since multiple starting values must be used.  Therefore, we took the following approach.

First, a collection $\{\theta_{k}:\, k = 1,\ldots,K_p\}$, of starting values was randomly generated for each setting $(n,p,\sigma^2)$, with the same starting values being used for all data sets under that setting.  The number of starting points was taken to be $K_2 = 10$, $K_5 = 50$, and $K_{10} = 100,$ so that these increase with the dimension $p.$  We then reduce this initial pool of starting values by optimizing a proxy to $W_n$ given by
\begin{equation}
    \label{eq:WnProxy}
    W_n^*(\theta) = n^{-1}\sum_{i = 1}^n d^2(Y_i, Y_i^*(h,\theta)),
\end{equation}
where
$$
Y_i^*(h,\theta) = \frac{\sum_{j = 1}^n r_h(X_j, X_i, \theta)Y_j}{\norm{\sum_{j = 1}^n r_h(X_j, X_i, \theta)Y_j}_E}
$$
is the projection onto the sphere of the local linear estimate of the Euclidean regression function $E(Y|\theta'X = u)$ at $u = \theta'X_i.$  The advantage of using this proxy is that an analytic gradient and Hessian for $W_n^*$ are available, so that optimization of $W_n^*$ is relatively fast. Using each of the $K_p$ starting values, we obtain as many initial estimates $\tilde{\theta}_k(h),$ $k = 1,\ldots,K_p.$ This optimization was executed using the \texttt{fmincon} function in Matlab with the \texttt{trust-region-reflective} option for the optimizer. In this optimization, $\theta$ was represented by its polar coordinates to handle the constraints in a simple way. 

In the final optimization step, $\tilde{K}_p$ of the initial estimates $\tilde{\theta}_k$ are retained as starting values based on having the lowest values of the proxy criterion $W_n^*$, with $\tilde{K}_2 = 2,$ $\tilde{K}_5 = 3,$ and $\tilde{K}_{10} = 5.$ For each starting value, $W_n$ is directly optimized using \texttt{fmincon} with the \texttt{SQP} option for the optimizer that does not require a gradient input, again using the polar representation of $\theta$.  The value of $\theta$ that, at convergence, attains the lowest value of $W_n$ is taken to be the estimate $\hat{\theta}(h)$ for that bandwidth.  Lastly, fitted values are computed using \eqref{eq:FSI_final} by setting $\tilde{h} = h$, $\hat{\theta} = \hat{\theta}(h),$ and $\hat{Y}_i(h) = \hmp(X_i).$

As a competitor to the FSI model, we also implemented a multivariate local Fréchet estimator. The estimator is defined as in \eqref{eq:fr_local_empirical}, with the only difference being that the weights $\hat{s}_h(X_i, x)$ are computed from multivariate local linear regression, since $X_i \in \mathbb{R}^p,$ using a product Gaussian kernel with the same bandwidth for each predictor.  The optimization for this estimator was performed using the ManOpt trust region algorithm described above.

\subsection{Performance Evaluation}

Data were generated under 18 unique parameter settings using samples sizes $n=50,100,200$, for noise levels $\sigma^2=0.4, 0.8$, and for dimensions $p=2,5,10$, with $200$ simulation runs per setting. Let $s=1,\ldots,200$ be the index for simulations within a given setting, and $(X_i^s,Y_i^s)$ denoted the simulated data.  Then, from each simulated data set and bandwidth we obtain an estimate $\hat{\theta}^{s}(h) \in \mathbb{R}^p$ and fitted values $\hat{Y}_{i}^{s}(h)$, $i = 1,\ldots,n$ from the FSI model, as well as fitted values $\check{Y}_i^s(h)$ from the multivariate local Fréchet (mLF) estimator. The following performance metrics were computed for each simulated data set across the entire range of bandwidths.

\begin{enumerate}
\item As the parameter space $\Theta_p$ is a subset of the $(p - 1)$-dimensional unit sphere, a natural measurement of empirical squared error for the $s$-th simulated data set is 
\begin{equation}
\operatorname{SE}(\hat{\theta}^s(h))= \left[{\arccos}\left(
 \Big| 
\theta_0'\,\hat{\theta}^{s}(h) \Big|  
\right)\right]^2,
\label{eq:MSE_theta}
\end{equation}
where we have introduced the absolute value to account for the fact that $\theta_0$ and $-\theta_0$ are indistinguishable from the data.
\item To evaluate the estimation error in regression for the FSI model, the mean square estimation error (MSEE) for the $s$-th simulated data set was quantified by 
\begin{equation}
\mathrm{MSEE}_{\oplus, \mathrm{FSI}}^{(s)}(h)= \frac{1}{n} \sum_{i=1}^{n}\left[\operatorname{arccos}\left(m_{\oplus}(X_i^s)' \hat{Y}_{i}^{s}( h)\right)\right]^{2}
\label{eq:MSEE_FSI}
\end{equation}

\item To evaluate the estimation error in regression for the multivariate local Fréchet estimator, the mean square estimation error (MSEE) for the $s$-th simulated data set was quantified by 
\begin{equation}
\mathrm{MSEE}_{\oplus,\mathrm{mLF}}^{(s)}(h)= \frac{1}{n} \sum_{i=1}^{n}\left[\operatorname{arccos}\left(m_{\oplus}(X_i^s)' \check{Y}_{i}^{s}( h)\right)\right]^{2}
\label{eq:MSEE_mLF}
\end{equation}

\end{enumerate}

Tables \ref{tab:sim_lownoise} and \ref{tab:sim_highnoise} show empirical performance metrics for the various simulation settings considered. In these tables, the average and standard deviation of each metric across simulations is reported.  For each metric, the reported values are for the bandwidth value in the chosen grid that minimizes the corresponding average across simulations. We observe that the average squared estimation errors and their standard deviations for the FSI estimator of the coefficient $\theta_0$, and both FSI and mLF estimators of the regression function $\mp(x)$, all behave in the expected fashion.  Namely, they decay toward zero with increasing sample size and are larger for higher values of $p$ and for the higher noise level.  However, the FSI regression estimation errors are overall smaller than those of the multivariate local Fréchet regression estimator when both are evaluated using their optimal bandwidth, with differences becoming more pronounced for larger covariate dimensions $p$.

\begin{table}[!ht]
\centering
\caption{\footnotesize 
Simulation results for settings with low noise, $\sigma^2=0.4$. Here $p$ and $n$ are covariate dimension and sample size, respectively. The third column is the average of the values $\mathrm{SE}(\hat{\theta}^s(h))$ from \eqref{eq:MSE_theta} across simulations, with standard deviation in parentheses.  Columns 4 and 5 give the averages of $\text{MSEE}_{\oplus,\mathrm{FSI}}^{(s)}(h)$ and $\text{MSEE}_{\oplus,\mathrm{mLF}}^{(s)}(h)$ from \eqref{eq:MSEE_FSI} and \eqref{eq:MSEE_mLF}, respectively, across simulations, with standard deviation given in parentheses.  For each of the metrics in columns 3--5, results are shown for the bandwidth that minimizes the reported average of that metric and are rounded to 3 significant digits.
}
\ \\
\begin{tabular}{|c|c|c|c|c|c|}
\hline 
&&&& \\
\ $p$\  & $n$ & \ \ \ \ \ Avg. $\mathrm{MSE}$ \ \ \ \ \ & \ \ \ Avg. $\mathrm{MSEE}_{\oplus,\mathrm{FSI}}$ \ \ \ & Avg. $\mathrm{MSEE}_{\oplus,\mathrm{mLF}}$ \\
&&&& \\
\hline 
&&&& \\
& 50 & 0.032 \ \ (0.047) & 0.063 \ \ (0.039) & 0.078 \ \  (0.042)\\
2 & 100 & 0.014 \ \ (0.020) & 0.030 \ \  (0.017) & 0.040 \ \  (0.020)\\
& 200 & 0.006 \ \ (0.008) & 0.016 \ \  (0.008) & 0.021 \ \  (0.010)\\
&&&& \\
\hline 
&&&& \\
& 50 & 0.326 \ \ (0.283) & 0.100 \ \  (0.051) & 0.143 \ \  (0.054)\\
5 & 100 & 0.168 \ \ (0.132) & 0.050 \ \ (0.026) & 0.074 \ \  (0.029)\\
& 200 & 0.071 \ \ (0.056) & 0.023 \ \  (0.012) & 0.036 \ \  (0.015)\\
&&&& \\
\hline 
&&&& \\
& 50 & 0.938 \ \ (0.519) & 0.166 \ \  (0.064) & 0.251 \ \  (0.081)\\
10 & 100 & 0.544 \ \ (0.386) & 0.082 \ \  (0.038) & 0.128 \ \  (0.038)\\
& 200 & 0.285 \ \ (0.152) & 0.039 \ \  (0.016) & 0.065 \ \  (0.018) \\
&&&& \\
\hline
\end{tabular}
\label{tab:sim_lownoise}
\end{table}
\begin{table}[!hb]
\centering
\caption{\footnotesize Simulation results for the settings with high noise, $\sigma^2=0.8$. Descriptions of column names and contents correspond to those given in Table~\ref{tab:sim_lownoise}.}
\ \\
\begin{tabular}{|c|c|c|c|c|c|}
\hline 
&&&& \\
\ $p$\  & $n$ & \ \ \ \ \ Avg. $\mathrm{MSE}$ \ \ \ \ \ & \ \ \ Avg. $\mathrm{MSEE}_{\oplus,\mathrm{FSI}}$ \ \ \ & Avg. $\mathrm{MSEE}_{\oplus,\mathrm{mLF}}$ \\
&&&& \\
\hline 
&&&& \\
 & 50 & 0.154  \ \  (0.284) & 0.231 \ \   (0.160) & 0.285 \ \   (0.177)\\
2 & 100 & 0.090 \ \   (0.244) & 0.130 \ \   (0.107) & 0.163 \ \   (0.100)\\
& 200 &  0.025 \ \   (0.037) & 0.063 \ \   (0.038) & 0.084 \ \   (0.048)\\
&&&& \\
\hline 
&&&& \\
& 50 & 1.038 \ \   (0.624) & 0.376 \ \   (0.180) & 0.558 \ \   (0.238)\\
5 & 100 &  0.680 \ \   (0.555) & 0.208 \ \   (0.118) & 0.298 \ \   (0.142)\\
& 200 & 0.367 \ \   (0.350) & 0.010 \ \   (0.056) & 0.143 \ \   (0.057)\\
&&&& \\
\hline 
&&&& \\
& 50 & 1.481 \ \   (0.528) & 0.496 \ \   (0.190) & 0.927 \ \   (0.302)\\
10 & 100 & 1.297 \ \   (0.578) & 0.298  \ \  (0.104) & 0.535 \ \   (0.171)\\
& 200 & 0.869 \ \   (0.477) & 0.160  \ \  (0.069) & 0.276 \ \   (0.083)\\
&&&& \\
\hline
\end{tabular}
\label{tab:sim_highnoise}
\end{table}

\clearpage

Next, we more closely examine the empirical sampling distribution of $\hat{\theta}(h)$ across different values of $n$ for $p=2$, since these can be easily visualized via histograms of the (scalar) polar coordinate representations $\hat{\eta}(h)$. Specifically, Figure \ref{fig:eta_hat_histogram_p2} shows the empirical distribution of $\hat{\eta}^{(s)}(h)$ for different values of $n$ and $\sigma^2$, where $h$ is the same minimizing bandwidth used to compute the average of the $\mathrm{SE}(\hat{\theta}^s(h))$ values for $p = 2$ in Tables~\ref{tab:sim_lownoise} and \ref{tab:sim_highnoise} for $\sigma^2 = 0.4$ and $\sigma^2 = 0.8,$ respectively.  For reference, the true polar coordinate parameter $\eta_0$ is superimposed as the red vertical line. In all cases, as $n$ increases the empirical sampling distribution becomes more concentrated near $\eta_0.$ 

\begin{figure}[!t]
\centering
\includegraphics[width=12cm]{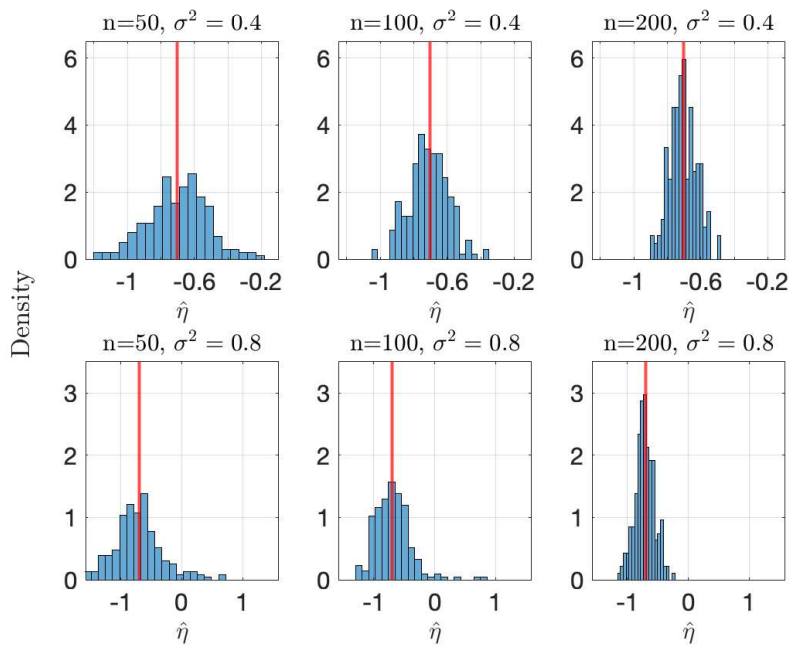}
\caption{\footnotesize For $p=2$ and sample sizes $n=50$ (left panels), $n=100$ (middle panels), $n=200$ (right panels) the simulated empirical distributions of $\hat{\eta}(h)$, the polar coordinate of $\hat{\theta}^s(h)$, are represented by histograms, with $h$ chosen to minimize the average of $\mathrm{SE}(\hat{\theta}^s(h))$ across simulations. In the top and bottom rows we have low noise ($\sigma^2=0.4$) and high noise ($\sigma^2=0.8$) scenarios respectively. The vertical red line represents the polar coordinate of $\theta_0,$ $\eta_0$ on the floor of the plot.}
\label{fig:eta_hat_histogram_p2}
\end{figure}

Finally, to more fully examine the estimation performance of the overall regression function $\mp(x)$ more closely, Figure \ref{fig:boxplots_msee_sphere} juxtaposes the boxplots of $\mathrm{MSEE}_{\oplus,\mathrm{FSI}}^{(s)}(h)$ from \eqref{eq:MSEE_FSI} for each simulation setting on the log scale, where $h$ is the minimizing bandwidth used for this metric in Tables~\ref{tab:sim_lownoise} and \ref{tab:sim_highnoise}. The variation increases with $p$, but under each $p$ it decreases with $n$. These reflect the numerical summaries given in Tables~\ref{tab:sim_lownoise} and \ref{tab:sim_highnoise}.

\begin{figure}[p]
\centering
\includegraphics[width=9.6cm]{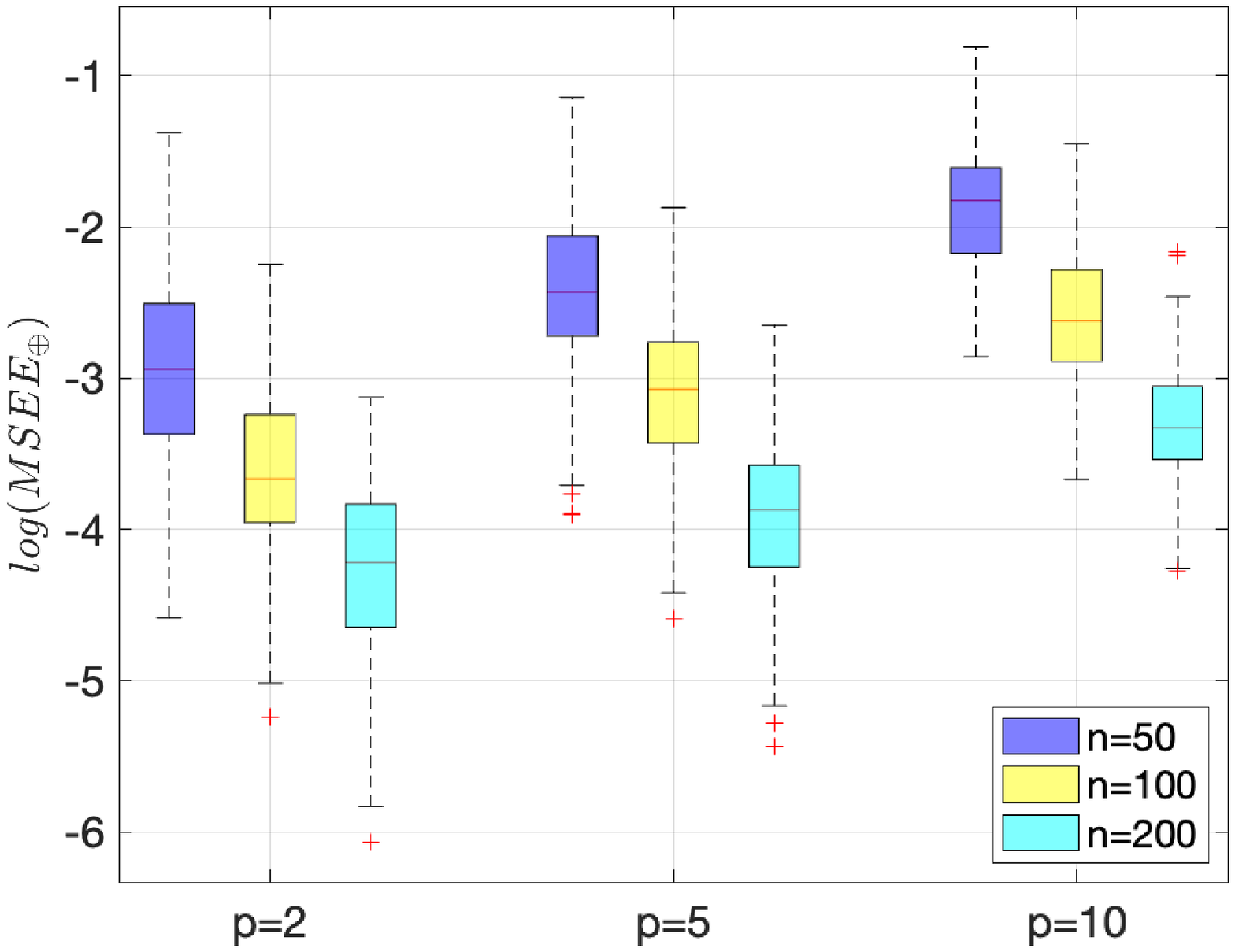} \qquad
\includegraphics[width=9.6cm]{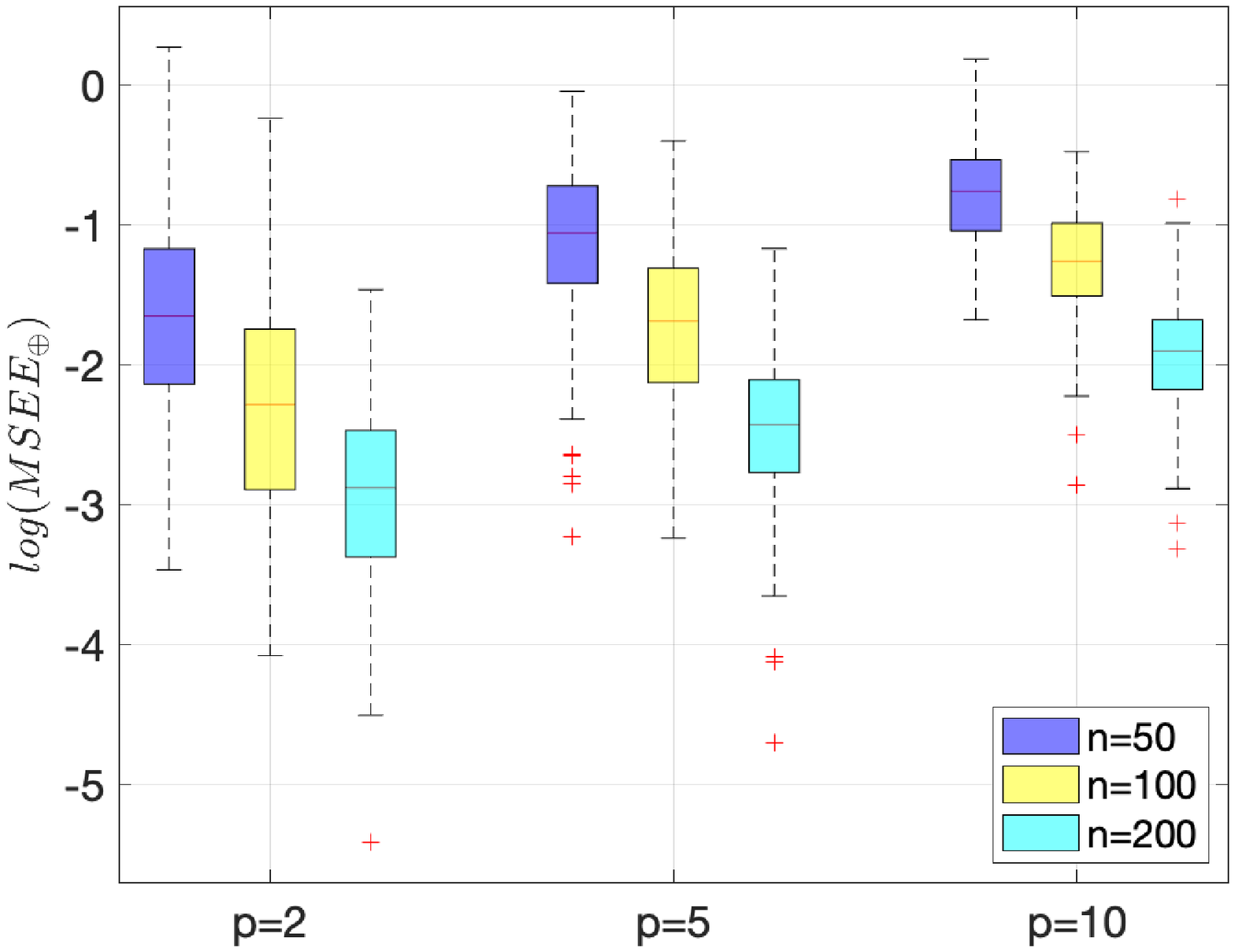}
\caption{\footnotesize For each covariate dimension  $p=2,5,10$; boxplots of $\log(\text{MSEE}_{\oplus_\mathrm{FSI}}^{(s)}(h))$ from \eqref{eq:MSEE_FSI} are given over all simulations for the optimizing bandwidths used in Tables~\ref{tab:sim_lownoise} and \ref{tab:sim_highnoise}, in each panel from left to right for sample sizes $n=50, 100, 200$ as indicated by blue, yellow, and 
cyan in the plot, respectively. The top and bottom panels correspond to low and high noise scenarios, respectively, with different vertical axis ranges.}
\label{fig:boxplots_msee_sphere}
\end{figure}

\vspace{0.5in}
\clearpage
\clearpage

\section{Regression of Mortality Distributions}
\label{sec:mort}

To demonstrate the application of our method, we consider human mortality data at the country level.  The goal is to model the dependence of age-at-death distributions
for a given year based on country-specific covariates. For this illustration, the year 2013 was selected, and human mortality data were sourced for 39 countries from the Human Mortality Database (HMD, \cite{HMD:2021} \href{https://www.mortality.org/}{\color{blue} {www.mortality.org}}) for this year. The HMD provides data for 41 countries; Hong Kong and Taiwan were omitted due to lack of availability of records for all covariates used in this illustrative example. The data for each country are structured as life-tables; for integer-valued age $j$, $0 \leq j \leq 110$, the life table provides the size of the population $m_j$ which is at least $j$ years old, normalized so that the total population is $m_0=100,000$. By computing differences, one can compute histograms of age-at-death that are specific to each country and year.  In order to focus on adult mortality, we consider the histogram over the age range $[20,110]$.  

The impacts of many socioeconomic, environmental, and other variables on health outcomes have been extensively researched.  For this illustration, we chose five covariates that, intuitively, have strong potential to influence mortality patterns of a nation.  These include year-on-year (YoY) percentage change in GDP (GDPC \cite{GDP2013}), carbon dioxide emissions in metric tons per capita (CO2E \cite{CO2E2013}), current health care expenditure as a percentage of GDP (HCE \cite{HCE2013}), the human development index (HDI \cite{HDI2013}), and infant mortality per 1000 live births (IM \cite{childMortalityData2013}) \citep{hassanipour2016incidence,ghoncheh2016incidence,rasoulinezhad2020mortality,granados2005recessions,eyer1977prosperity,higgs1979cycles,graham1992poorer,owusu2021relationship,lippi2016no,rothberg2010little}. Hence, $X_i \in \R^5$ constitutes the covariate vector for the $i-$th country, $i=1,\ldots,39$. 

To apply the proposed FSI model, the density histograms constructed from the lifetables were smoothed and then used to produce a quantile function for each country.  This smoothing step was performed using the \texttt{CreateDensity} function in the R package \texttt{frechet} in order to obtain a smooth density, with the default cross-validated bandwidth choice, followed by conversion to a quantile function using the function \texttt{dens2quantile} in the package \texttt{fdadensity} \citep{frechet,fdadenspackage}.  These constructed distributions will be referred to as observed distributions, and are visualized in Figure~\ref{fig:dens_2013_estimation}.  
\begin{figure}[!ht]
\centering
\includegraphics[width = 9.6cm]{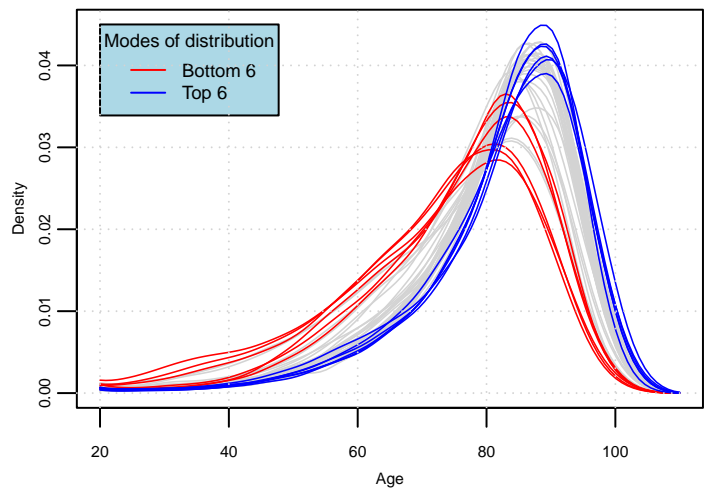}
\caption{\footnotesize The estimated densities for each country for year 2013 over the age interval [20,110]; the countries with top 6 and bottom 6 mode ages are highlighted in blue and red colors respectively. The red colored densities include Russian Federation, Belarus, Ukraine, Hungary, Slovakia, Latvia. The densities in blue include Australia, Canada, Spain, France, Japan, Switzerland.}
\label{fig:dens_2013_estimation}
\end{figure} 

Let $Y_i$ represent the observed age-at-death distribution with quantile function $q_i$ for the $i$-th country, and $X_i$ a vector of covariates, $i = 1,\ldots,40$, during 2013.  The random object responses $Y_i$ are assumed to belong to the space $\Omega$ of probability distributions $G$ on $\mathbb{R}$ with finite second moment, i.e. $\int_{\mathbb{R}}x^2 \,dG(x) < \infty$. For two distributions $G_1,G_2 \in \Omega$, the squared Wasserstein distance \citep{vill:03} between them is
\begin{equation}
  d_{W}^2\left(G_{1}, G_{2}\right)=\int_{0}^{1}\left(G_{1}^{-1}(t)-G_{2}^{-1}(t)\right)^{2} \mathrm{d}t,
  \label{eq:wass_metric}
\end{equation}
where $G_{1}^{-1},G_{2}^{-1}$ are the quantile functions corresponding to $G_1,G_2$ respectively. The above form of the metric makes obvious the point raised previously that the Wasserstein space is isometric to a subset of the Hilbert space $L^2[0,1]$.  Thus, it is a flat Hadamard space, though it is convex and not linear.  While one may, to some extent, employ linear methods to analyze such data, practical and theoretical problems emerge even in this simple case.  From a practical standpoint, certain critical outputs, such as fitted values, that should be distribution-valued may not be so when linear methods are applied.  These may be easily remedied using an ad hoc correction, but this is a clear disadvantage compared to the object treatment provided by Fréchet methods that will always respect such constraints.  Beyond estimation, use of the non-linear geometry has distinct advantages when it comes to inference, particularly in the formulation of error models and uncertainty assessment, even in the setting of univariate distributions \citep{panaretos2016amplitude,petersen2019wasserstein,petersen2021wasserstein}.  In addition, although univariate distributions are employed in this illustrative example, the model is equally applicable to multivariate distributions \citep{zemel2019frechet}, in which case the Wasserstein space is no longer flat.

Letting $(X,Y)$ denote a generic covariate-distribution pair, the target is the Fréchet regression function $\mp$ as defined in \eqref{eq:fr_regressionfunction}, for which we will assess seven competing models for object data. Specifically, $\mp$ was estimated using global and local Fréchet regression techniques, the latter for each individual predictor, yielding six competitors to the proposed FSI model in \eqref{eq:fsi_model}.  

\subsection{Computational Details}
\label{ss:mortComp}

The computations for global and local Fréchet estimates, the letter for any fixed bandwidth, were carried out using the existing functionalities of the \texttt{frechet} package \cite{frechet}.  For the FSI model, for any specified $\theta$ and bandwidth $h$, this package was also used to compute $\hgp$ in \eqref{eq:gplus_est}.  To estimate $\theta_0$ via \eqref{eq:theta_est}, the \texttt{optim} command was used with option \texttt{"L-BFGS-B"} \cite{byrd1995limited} with a lattice of $3^4=81$ starting points of polar coordinates $\eta \in [-\pi/2,\pi/2]^4$. The predictors were each centered and scaled to have sample mean zero and unit sample variance prior to fitting all models. For simplicity we use the same acronyms for the standardized covariates as previously given for the unstandardized ones, with the $X_i$ values in each model being on the standardized scale.

As a first step, for each of the local Fréchet regression fits and the FSI model fit, a single bandwidth was selected by leave-one-out cross validation on the entire data set; no bandwidth is needed for global Fréchet regression. With $m$ denoting a model index corresponding to the FSI model or one of the local Fréchet fits, let $\hat{Y}_{i}^{(m, -i)}(h)$ denote the fitted value for the $i$-th country produced by the estimate of model $m$ using all countries except county $i$ and with bandwidth $h.$  Then the chosen bandwidth is
\begin{equation}
    \label{eq:hopt_loocv}
    h^*_m = \mathop{\argmin}\limits_{h \in \mathfrak{H}_m} \ \sum_{i = 1}^n d_W^2(Y_i, \hat{Y}_{i}^{(m,-i)}(h)),
\end{equation}
where $\mathfrak{H}_m$ is a grid of potential bandwidth choices for the given model.  For the local Fréchet fits of each individual predictor, this step was executed using built-in functionalities of the \texttt{frechet} package.  For the FSI bandwidth, the model was fit for each bandwidth in a pre-defined grid as described above, then $h^*_{\mathrm{FSI}}$ was computed as in \eqref{eq:hopt_loocv}.

\subsection{Model Comparisons}
\label{ss:mortModels}

To assess model performance, two metrics were computed.  The first metric, termed the Fréchet $R^2$, quantifies the quality of model fit by in-sample performance.  Specifically, for a given model $m$, let $\hat{Y}_i^{(m)}$ denote the fitted value that it produces for the $i$-th country.  Furthermore, let 
$$
\hat{\omega}_\oplus = \mathop{\argmin}\limits_{\omega \in \Omega} \, \frac{1}{n} \son d_W^2(Y_i, \omega)
$$ 
denote the sample Fréchet mean.  Indeed, this is simple to compute due to the nature of $d_W$ in \eqref{eq:wass_metric}, as it is known that $\hat{\omega}_\oplus$ is the distribution with quantile function $n\inv\son q_i.$  The Fréchet $R^2$ for model $m$ is
\begin{equation}
\label{eq:fr_rsquared}
R_{\oplus,m}^2=1-\frac{\sum_{i=1}^{n}d_W^2(Y_i,\hat{Y}_i^{(m)})}{\sum_{i=1}^{n}d_W^2(Y_i,\hat{\omega}_{\oplus})},
\end{equation}
which measures the proportion of Wasserstein-Fréchet variability in the data that is explained by the model. 

The second performance metric is based on out-of-sample performance, in which the data were randomly split into a testing set of size 10 and training set of size 29, with 30 distinct random splits being executed.  With $k=1,\ldots,30$ representing the index of each unique split of the data, denote by $Y_{[k,j]},$ $j=1,\ldots,10,$ the age-at-death distribution for the $j$-th country in the $k$-th testing set, and by $\hat{Y}_{[k,j]}^{(m)}$ the predicted distribution for the same country using the fit of model $m$ produced by the $k$-th training set.  The error for the $k$-th split and model $m$ is then quantified by
\begin{equation}
\label{eq:fr_mspe}
\mathrm{MSPE}_k^{(m)} = \frac{1}{10}\sum_{j=1}^{10}d_W^2\left(Y_{[k,j]},\hat{Y}_{[k,j]}^{(m)}\right).
\end{equation}
For local Fréchet and FSI model fits, the bandwidth used for each training set was fixed to be the value $h^*_m$ in \eqref{eq:hopt_loocv}.

\begin{table}[!t]
\centering
\caption{\footnotesize Performance metrics for comparing seven Fréchet regression fits in three classes of models: (GF) Global Fréchet, (LF) local Fréchet, (FSI) Fréchet single index. The predictor used for each local Fréchet fit is indicated for each subcolumn below LF: (HDI) Human Development Index; (HCE) current health care expenditure as a percentage of GDP; (GDPC) GDP year-over-year growth percentage; (IM) infant mortality; (CO2E) carbon dioxide emissions. The $R_{\oplus}^2$ row gives the Fréchet $R^2$ values defined in \eqref{eq:fr_rsquared}.   The MSPE row gives the average out-of-sample mean-square prediction error ($\mathrm{MSPE}$), defined in \eqref{eq:fr_mspe}, across the 30 data splits. The SD(MSPE) row gives the standard deviation of the out of sample prediction errors across the 30 data splits.}
\resizebox{12cm}{!}{
\begin{tabular}{|l||c||c|c|c|c|c||c|}
\hline  &  & \multicolumn{5}{c||}{ } &  \\
 Evaluation & & \multicolumn{5}{c||}{LF} &  \\
\cline{3-7}
Measures & GF &&&&&& FSI \\
  &  & HDI & HCE & GDPC & IM & CO2E & \\
\hline &&&&&&& \\
$R_{\oplus}^2$  & 0.697 & 0.688 & 0.521 & 0.132 & 0.433 & 0.162 & 0.827 \\  
&&&&&&& \\ \hline
&&&&&&& \\
$\mathrm{MSPE}$ & 6.23 & 6.87 & 6.93 & 13.51 & 12.08 & 13.74 & 4.35 \\ 
&&&&&&& \\ 
\hline
&&&&&&& \\
$\mathrm{SD(MSPE)}$ & 2.45 & 5.45 & 3.44 & 4.82 & 10.03 & 5.40 & 2.11 \\ 
&&&&&&& \\
\hline
\end{tabular} }
\label{tab:mspe_rsq_oplus_compare}
\end{table}

\begin{figure}[ht!]
\centering
\includegraphics[width=12cm]{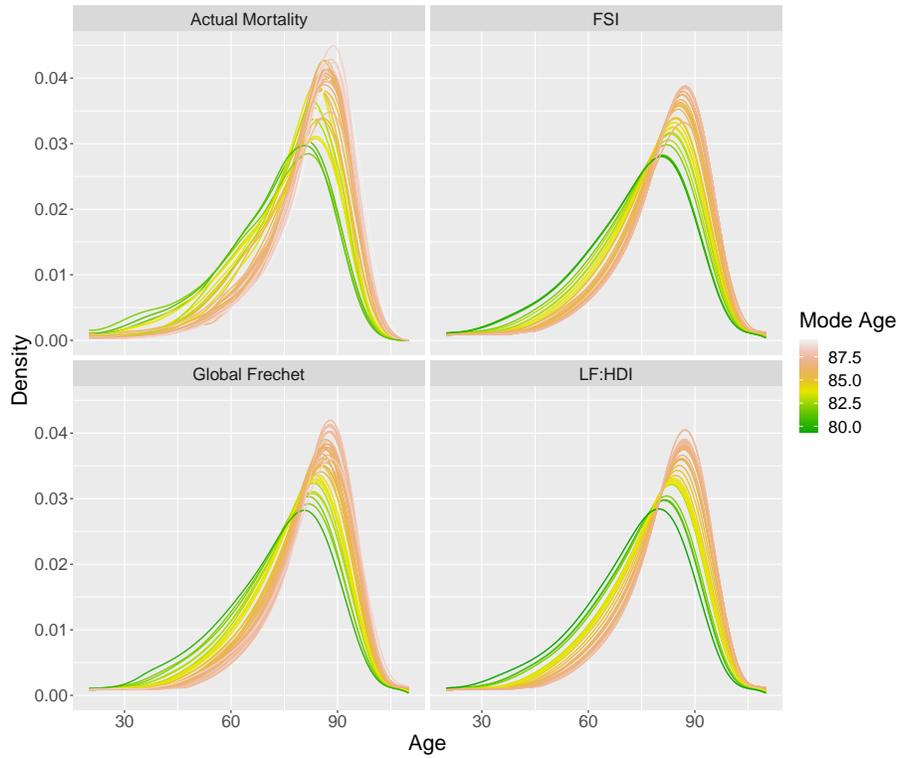}
\caption{\footnotesize Observed smooth densities (top left) along with their fits produced by the proposed FSI model (top right), Global Fréchet model (bottom left), and local Fréchet regression with HDI as predictor (bottom right). Densities are colored by the mode of the age-at-death distribution.}
\label{GF_LF_in_sample_pred.}
\end{figure} 

Table~\ref{tab:mspe_rsq_oplus_compare} gives the computed metrics for all models.  The top three models in terms of Fréchet $R^2$ are the proposed FSI model, the local Fréchet fit using the HDI covariate, and the global Fréchet model.  Figure~\ref{GF_LF_in_sample_pred.} plots the fitted distributions (as densities) for these three models, along with the observed densities.  The plot provides a visual reinforcement of the Fréchet $R^2$ findings as these three models all produce distribution fits that approximate the observed distributions reasonably well.  

Using out-of-sample performance, the FSI model emerges as the best model with the lowest average $\mathrm{MSPE}$ of 4.35.  The left panel of Figure~\ref{fig:boxplots_models_compare} shows boxplots of the 30 different $\mathrm{MSPE}_{k}^{(m)}$ values for each model across splits, reinforcing the metrics in Table~\ref{tab:mspe_rsq_oplus_compare}. In addition to having the smallest median $\mathrm{MSPE}$ value, the dispersion across folds for the FSI is among the lowest, second only to the global Fréchet model.  The global Fréchet model suffers from model-induced bias, while the local Fréchet estimates using HDI lack relevant information from other variables and suffer from poor prediction in certain data splits.  As designed, the FSI model balances the strengths of these two models.  However, these results do not examine the relative performance of these models for each individual split of the data. The right panel of Figure~\ref{fig:boxplots_models_compare} shows the boxplots of the logarithm of the ratio of MSPEs for each of three competing models (global Fréchet and local Fréchet estimates using HDI and HCE, respectively) to the MSPEs of FSI across splits. This comparison shows FSI as the best in overall out-of-sample prediction, as its prediction error is smaller than that of the other top-performing models for the majority of the 30 training/test data splits.

Next, we intepret the coefficient estimate for the FSI model.  Rounded to three digits after the decimal, this was
$$\hat{\theta}=(0.667, \, 0.741, \ -0.067, \, 0.005, \,  0.046)'.$$ 
with the order of standardized covariates being Human Development Index (HDI), Healthcare expenditure as percentage of GDP (HCE), year-on-year percentage change in GDP (GDPC), infant mortality per 1000 live births (IM), carbon dioxide emissions metric tonnes per capita (CO2E).The estimated coefficients for HDI and HCE have the highest magnitudes of 0.667 and 0.741 respectively, indicating their heavy influence relative to the other three predictors on the index $\hat{U}_i = \hat{\theta}'X_i$ that drives the FSI fit, when all variables are in the model. As the FSI fit can be viewed as a local Fréchet estimate based on the univariate predictor $\hat{U}_i$, the superiority of the FSI model to the local Fréchet fit using either the HDI or HCE as predictor indicates that the combined predictive power of HDI and HCE, as quantified by the projection direction $\hat{\theta}$, is stronger than either individual predictor when using local Fréchet regression.  On the other hand, the global Fréchet model also combines the influence of all predictors, but does so less efficiently due to bias in the underlying model.  
\begin{figure}[t!]
\centering
\includegraphics[scale = 0.25]{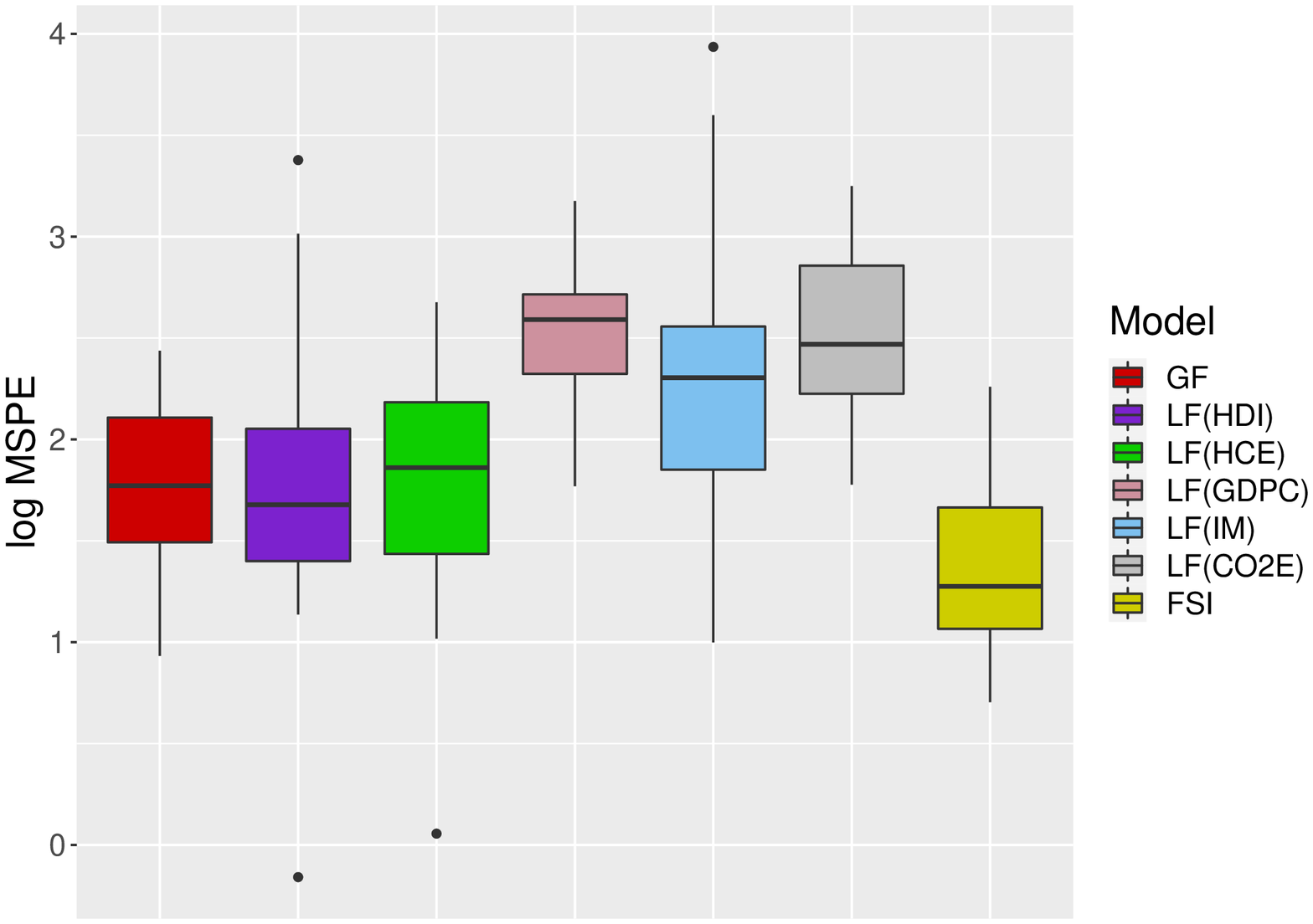}
\includegraphics[scale = 0.25]{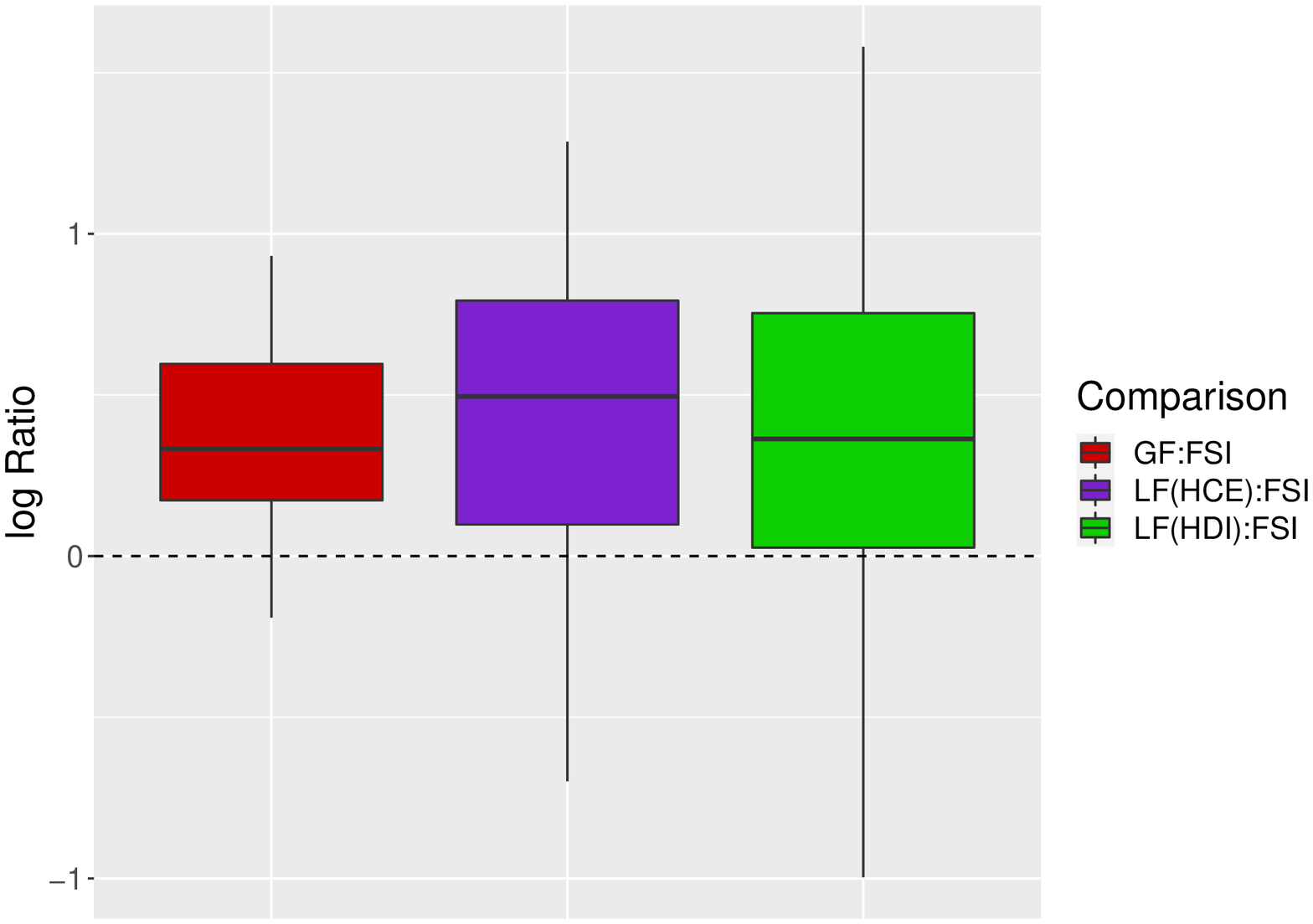}
\caption{\footnotesize Left panel: boxplots of $\text{MSPE}_k^{(m)}$ values from \eqref{eq:fr_mspe} across splits for the following estimates from left to right: Global Fréchet (GF); local Fréchet for each of the predictors Human Development Index (HDI), Healthcare expenditure as \% of GDP (HCE), GDP YoY\% change (GDPC), Infant Mortality per 1000 live births (IM), and $\text{CO}_2$ emissions metric tonnes per capita (CO2E); and Fréchet Single Index (FSI). Right panel: boxplots of log of ratio of the MSPEs from global Fréchet estimates (dark red, left), local Fréchet estimates using HDI (dark purple, midde), and local Fréchet estimates using HCE (green, right) to those of FSI are shown. The MSPE values of each competitor are higher than the FSI values for more than 75\% of the folds, shown by the first quartile of the log-ratios being above the dotted horizontal line.}
\label{fig:boxplots_models_compare}
\end{figure} 

Since HDI and HCE appear to have relatively higher importance as predictors of mortality distributions for the local Fréchet regression as well as for the FSI model in terms of both in-sample and out-of-sample performance, it was interesting to explore how a small change in standardized value of HDI or HCE would affect the mortality distribution prediction of FSI model, while keeping all other covariates fixed at their median values. Figure \ref{fig:HDI_effect} shows the age-at-death distributions predicted by the fitted FSI model. As expected, higher HDI or HCE are associated with increased longevity.
In particular, the plots suggests that the mode of mortality distributions increases for higher values of HDI or HCE, keeping other covariates fixed.

\begin{figure}[ht!]
\centering
\includegraphics[width=5.9cm]{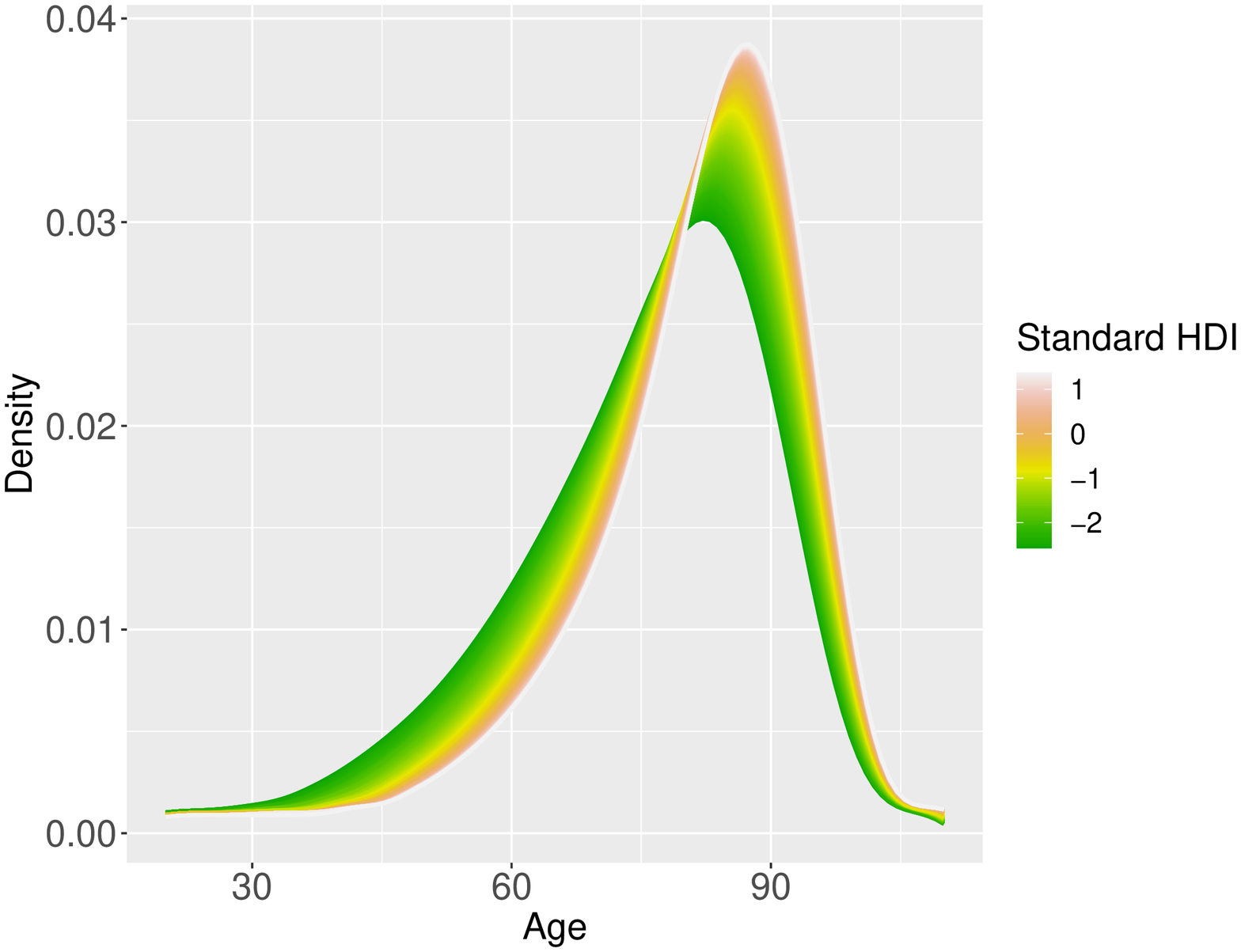}
\includegraphics[width=5.9cm]{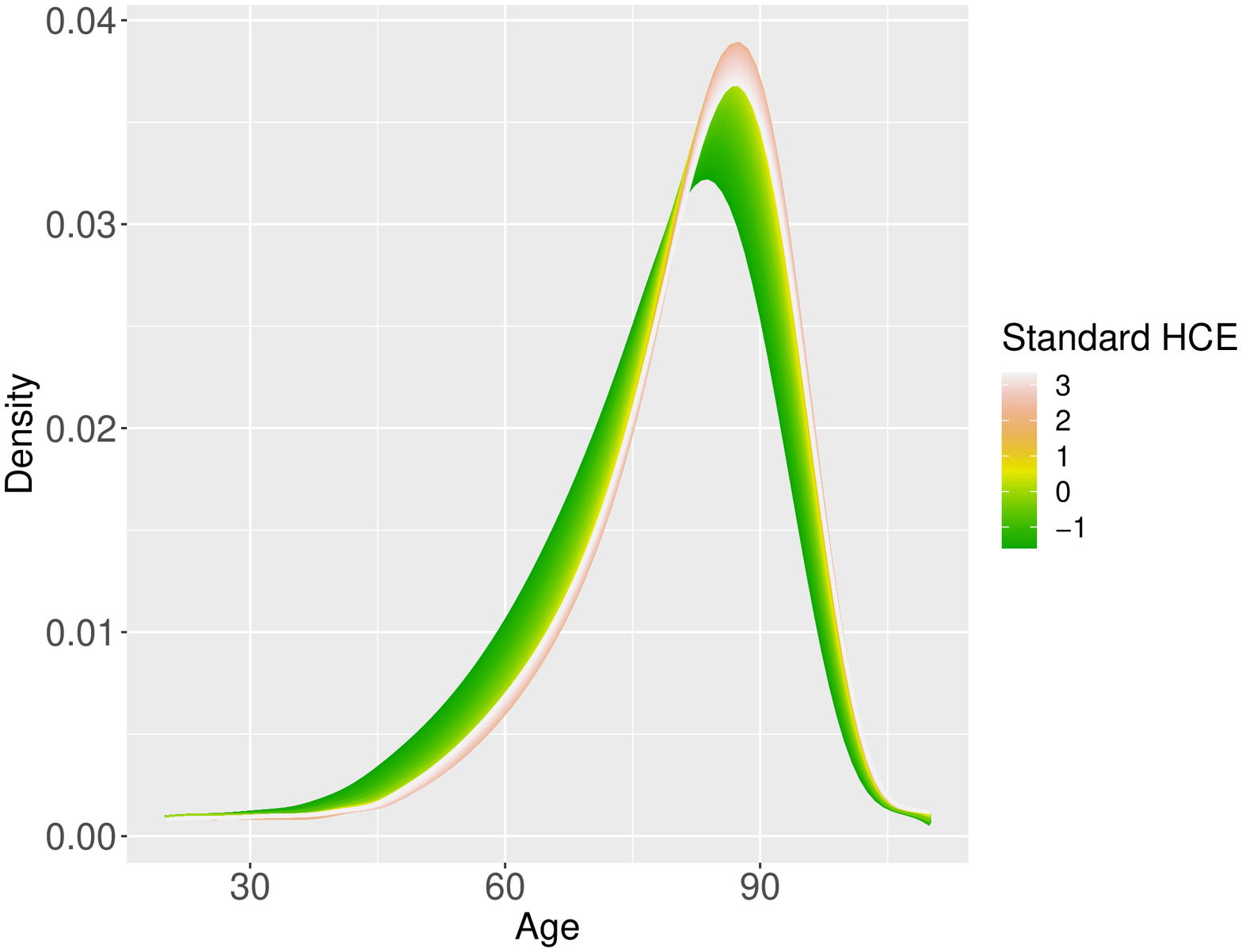}
\caption{\footnotesize Fitted age-at-death densities produced by the FSI model for varying values of HDI (human development index, left) and HCE (health care expenditure, right), with other variables set to their sample median. Colors indicate regularly spaced standardized values of the covariate.}
\label{fig:HDI_effect}
\end{figure}

\section{Discussion}
\label{sec:disc}

The Fréchet single index model developed in this paper offers an alternative to global and local Fréchet regression for random object response data with vector-valued predictors in the spirit of semiparametric regression.  While global Fréchet regression comfortably accommodates multiple predictors, it can be unduly rigid for many complex data settings.  Indeed, even in the special case $\Omega = \R,$ in which global Fréchet is multiple linear regression, such a model is often inadequate, so that its inadequacy in more complex metric spaces $\Omega$ is more likely than not.  Local Fréchet regression, on the other hand, is unattractive when multiple predictors are present on both theoretical and practical grounds, despite its flexibility.  Indeed, the data illustration involving mortality profiles demonstrates that the FSI model outperforms both global Fréchet regression and the best single-predictor model fitted using local Fréchet regression.  Future extensions of the FSI model to handle more complex predictors, such as high-dimensional, functional, or object-valued data, will be valuable assets.

The technical issue surrounding existence and uniqueness of Fréchet means, whether marginal or conditional, has been circumvented in this work by assumption, although specific concrete examples of spaces satisfying the relevant assumption (M) have been provided due to the work of others on this challenging topic.  Nevertheless, as pointed out by reviewers, a particular limitation of the FSI model is its requirement that the conditional Fréchet means $\mp(x)$ in \eqref{eq:fr_regressionfunction} not only exist for each $x$, but that those conditional on $\theta'x$, namely $g_\oplus(\theta'x,\theta)$ in \eqref{eq:gplus}, exist and are unique for every $\theta.$ Examples can be quickly constructed in which the FSI model holds while $g_\oplus(\theta'x,\theta)$ are only unique for $\theta$ equal to or in a neighborhood of $\theta_0$.  It seems plausible that one should still be able to estimate $\theta_0$ in this setting, yet the methods proposed in this paper are inadequate.  It is likely that criterion functions less restrictive than \eqref{eq:W} may provide a path, and we leave this for future work.

While we have used a generalized version of semiparametric least squares for the estimation of the coefficient vector and local Fréchet regression to estimate $g_\oplus$ in \eqref{eq:gplus}, other options are of course available.  For example, projection pursuit \citep{frie:81,hall1989projection}, average derivatives \citep{hardle1989investigating}, the conditional minimum average variance estimation (MAVE) technique \citep{xia2006asymptotic}, and sliced inverse regression \citep{li1991sliced}, among others, have been validated practically and theoretically for scalar responses.  Such approaches could conceivably work for object responses as alternatives to the method presented here for estimating the coefficient in the FSI model.  More broadly, alternative smoothing methods could be developed for the estimation of the link function $\gp,$ although local Fréchet regression and the Nadaraya-Watson estimator \citep{hein2009robust} seem to be the only available options to date for a general metric space.  Other semiparametric approaches for scalar data, such as multiple index models, may well prove to be adaptable to this scenario, although their extensions are less obvious.  

\section{Proofs}
\label{sec:proofs}

This section provides detailed arguments for establishing the main results in Section~\ref{ss:theory}.  First, Section~\ref{ss:aux} contains proofs of Lemmas~\ref{lma:bias}--\ref{lma:BB_approx2}.  These lemmas contain the necessary pieces for establishing Theorem~\ref{thm:unif_cons} and Corollary~\ref{cor:theta_cons}, which are proved in Section~\ref{ss:main_proofs}.

\subsection{Lemmas~\ref{lma:bias}--\ref{lma:unifKernReg} and their proofs}
\label{ss:aux}
Define
\begin{equation}
    \label{eq:mujetc}
    \begin{split}
    \mu_{j,\theta}(x) &= E\left[K_h(\theta'(X_1 - x))(\theta'(X_1 - x))^j\right], \, j = 0,1,2,  \\
    \sigma_{\theta}^2(x) &= \mu_{0,\theta}(x)\mu_{2,\theta}(x) - \mu_{1,\theta}^2(x), \\
    r_h(z, x, \theta) &= \sigma_{\theta}^{-2}(x)K_h(\theta'(z - x))\left[\mu_{2,\theta}(x) - \mu_{1,\theta}(x)(\theta'(z - x))\right], \\ \tilde{\Lambda}_\oplus(\omega, \theta'x, \theta) &= E[r_h(X, x, \theta)d^2(Y,\omega)]. 
    \end{split}
\end{equation}
Lemma~\ref{lma:bias} will establish the uniform rate of convergence of $\tilde{\Lambda}_\oplus$ to the population target as a function of the bandwidth $h$.  Lemma~\ref{lma:BB_approx} will then utilize a multivariate Brownian bridge approximation to establish rates of convergence of kernel estimates for the densities of projections $\theta'X$ that are uniform in both $\theta$ and across the density argument.  In turn, Lemma~\ref{lma:unifKDE} applies the result of Lemma~\ref{lma:BB_approx} to the components $\hmu_{j,\theta}$ that are used in the local Fréchet estimator in \eqref{eq:fsi_weights}.  Finally, Lemmas~\ref{lma:BB_approx2} and \ref{lma:unifKernReg} give the details of a second multivariate Brownian bridge approximation as a generalization of the results in \cite{mack1982weak}.

\begin{lemma}
\label{lma:bias}
Under assumptions (K) and (F),
$$
\sup_{\omega,\theta,x} |\tilde{\Lambda}_\oplus(\omega,\theta'x,\theta) - \Lambda_\oplus(\omega,\theta'x,\theta)| = o_p(1).
$$
\end{lemma}

\begin{proof}
First, we follow the steps of Theorem~3 in \cite{petersen2019frechet} to establish that
\begin{equation}
\label{eq:absCont}
\frac{\d F_{Y|\theta'X}(y|v)}{\d F_Y(u)} = \frac{f_{\theta'X|Y}(v|y)}{f_{\theta'X}(v)}.
\end{equation}
Let $V \subset \Omega$ be any open set.  For $\theta \in \Theta_p$ and $u \in \mc{U}_\theta,$ set
  $$
  a_\theta(u) = \int_V \frac{f_{\theta'X|Y}(u|y)}{f_{\theta'X}(u)} \mathrm{d} F_{Y}(y), \quad b_\theta(u) = \int_V \mathrm{d}F_{Y|\theta'X}(y|u).
  $$
  By Tonelli's theorem, for any $s \in \mathbb{R},$
  \begin{equation*}
    \begin{split}
    \int_{-\infty}^s a_\theta(u)f_{\theta'X}(u) \d u &=  \int_V\left[\int_{-\infty}^s f_{\theta'X|Y}(u|y)\d u\right] \d F_Y(y) \\
    &= \int_{-\infty}^s \left[\int_V \d F_{Y|\theta'X}(y|u)\right]f_{\theta'X}(u)\d u \\
    &= \int_{-\infty}^s b_\theta(u) f_{\theta'X}(u)\d u.
    \end{split} 
  \end{equation*}
Hence, under part iii) of (F), it follows that $a_\theta(u) = b_\theta(u),$ whence \eqref{eq:absCont} follows. Furthermore, let
\begin{equation}
    \label{eq:phiDef}
    \phi(\omega,u,\theta) = \frac{\partial^2}{\partial u^2} \Lambda_\oplus(\omega,u,\theta),
\end{equation}
which is guaranteed to exist for any $\omega \in \Omega,$ $\theta \in \Theta_p$ and $u \in \mc{U}^\mathrm{o}$ by part iii) of (F) and dominated convergence.  Indeed, (F) ensures that
\begin{equation}
    \label{eq:phiBound}
\left|\phi(\omega,u,\theta)\right| = \left|\int_\Omega d^2(y,\omega)\frac{\partial^2}{\partial u^2}\left[\frac{f_{\theta'X|Y}(u|y)}{f_{\theta'X}(u)}\right] \d F_Y(y)\right| < \infty
\end{equation}
uniformly in $\omega,$ $\theta,$ and $u \in \mc{U}_\theta.$

Next, using both (K) and part ii) of (F), we can follow the arguments in the proof of Theorem~1 in \cite{chen2022uniform} to obtain the expansion
\begin{equation}
    \label{eq:tLambdaExp}
\tilde{\Lambda}_\oplus(\omega,\theta'x,\theta) = \Lambda_\oplus(\omega,\theta'x,\theta) + E\left(r_h(X, x, \theta)\phi(\omega,\theta'X^*,\theta)\left[\theta'(X - x)\right]^2\right),
\end{equation}
where $X^*$ is some intermediate value between $X$ and $x$. Using (K), (F), and basic results of \cite{fan1996local},
\begin{equation}
    \label{eq:rhBound}
    E\left(|r_h(X, x, \theta)|\left[\theta'(X - x)\right]^2\right) = O(h^2)
\end{equation}
uniformly in $x$ and $\theta.$ Applying \eqref{eq:phiBound} and \eqref{eq:rhBound} to \eqref{eq:tLambdaExp}, the result follows.
\end{proof}

\begin{lemma}
\label{lma:BB_approx}
Let $\kappa:\R \rightarrow \R$ satisfy the following properties:
\begin{itemize}
    \item $\kappa$ is uniformly continuous and of bounded variation;
    \item $\int_\R |\kappa(w)|\mathrm{d}w < \infty;$
    \item $\kappa(u) \rightarrow 0$ as $|u| \rightarrow \infty;$ and
    \item $\int_\R|u\log|u||^{1/2}\mathrm{d}\kappa(u) < \infty.$
\end{itemize}
For $h > 0$,  $x$ in the support of $F_X,$ and $\theta \in \R^p$ with $\lVert \theta \rVert = 1,$ set
\[    
    \hat{\nu}_\theta(x) = n\inv\son \kappa_h(\theta'(X_i - x)),
\]
where $\kappa_h(\cdot) = h\inv\kappa(\cdot/h).$  Then, on a sufficiently rich probability space, there exist processes $\rho_{n,\theta}(x)$ and $\varepsilon_{n,\theta}(x)$ such that, provided $h \rightarrow 0 $ as $n \rightarrow \infty,$
\begin{equation}
\label{eq:BB_approx}
\begin{split}
    \hat{\nu}_\theta(x) &= E(\hat{\nu}_\theta(x)) + n^{-1/2}\rho_{n,\theta}(x) + n^{-1/2}\varepsilon_{n,\theta}(x), \\
    \sup_{x,\theta} |\rho_{n,\theta}(x)| &= O_P\left\{(-h\inv \log h)^{1/2}\right\}, \, \text{and} \\
    \sup_{x,\theta} |\varepsilon_{n,\theta}(x)| &= O\left(\frac{(\log n)^{3/2}}{hn^{1/[2(p+1)]}}\right) \, \textrm{a.s.} 
\end{split}
\end{equation}

\end{lemma}

\begin{proof}
  Let $T$ be the Rosenblatt transformation defind in assumption (F) point i), and let $\hat{G}$ be the empirical distribution function of $T_i = T(X_i)$, $i = 1,\ldots,n.$  \cite{csorgHo1975new} proved the existence of a sequence of Brownian bridges $B_n$ (with continuous sample paths) on the $p$-dimensional unit cube such that
  \begin{equation}
      \label{eq:csorgho}
      \sup_{t \in [0,1]^p} |n^{1/2}(\hat{G}(t) - t) - B_n(t)| = O\left(\frac{(\log n)^{3/2}}{n^{1/[2(p+1)]}}\right) \, \text{a.s.} 
  \end{equation}
  Define $Z_n(x) = n^{1/2}(\hat{F}_X(x) - F(x)) - B_n(T(x))$, where $\hat{F}_X$ is the empirical distribution function of the $X_i$, and let $F_{\theta'X}$ and $\hat{F}_{\theta'X}$ be, respectively, the population and empirical distribution functions of $\theta'X_i,$ so that
\begin{equation*}
  \begin{split}
  \hat{F}_{\theta'X}(u)  - F_{\theta'X}(u) &= \int_{\mathbb{R}^p} \mathbf{1}(\theta'z \leq u) \d\left[\hat{F}_X(z) - F_X(z)\right] \\
  &= n^{-1/2} \int_{\mathbb{R}^p} \mathbf{1}(\theta'z \leq u) \d B_n(T(z)) \\
  &\hspace{1cm} + n^{-1/2} \int_{\mathbb{R}^p} \mathbf{1}(\theta'z \leq u) \d Z_n(z) \\
  &= n^{-1/2}B_{n,\theta}(u) + n^{-1/2}Z_{n,\theta}(u).
  \end{split}
 \end{equation*}
  Hence,
  \begin{equation}
      \label{eq:nu_exp}
      \begin{split}
        \hat{\nu}_\theta(x) &= E(\hat{\nu}_\theta(x)) + \int_{\mathbb{R}} \kappa_h(u - \theta'x)) \mathrm{d}(\hat{F}_{\theta'X} - F_{\theta'X})(u) \\
        &= E(\hat{\nu}_\theta(x)) + n^{-1/2}\rho_{n,\theta}(x) + n^{-1/2}\varepsilon_{n,\theta}(x),
      \end{split}
  \end{equation}
  where $\rho_{n,\theta}(x) = \int_\mathbb{R} \kappa_h(u-\theta'x) \mathrm{d}B_{n,\theta}(u)$ and $\varepsilon_{n,\theta}(x) = \int_\mathbb{R} \kappa_h(u - \theta'x) \mathrm{d}Z_{n,\theta}(u).$  Thus, we have established the first line of \eqref{eq:BB_approx}.  Due to conditions on the kernel and \eqref{eq:csorgho}, the third line of \eqref{eq:BB_approx} immediately follows.
  
  To establish the second line of \eqref{eq:BB_approx}, use integration by parts, the change of variable $s = (u - \theta'x)/h$, and the assumption that $\kappa(u) \rightarrow 0$ as $|u| \rightarrow 0$ to conclude that
\begin{equation}
    \label{eq:rhoBound1}
    \begin{split}
        |\rho_{n,\theta}(x)| &= \left|\int_{\mathbb{R}} \kappa_h(u - \theta'x) \d B_{n,\theta}(u)\right| \\
        &\leq h^{-1}\int_{\mathbb{R}} |B_{n,\theta}(\theta'x + hs)||\d \kappa(s)| \\
        &\leq h^{-1}\int_\mathbb{R} |B_{n,\theta}(\theta'x + hs) - B_{n,\theta}(\theta'x)||\d \kappa(s)|.
    \end{split}
\end{equation}
  To control the integrand, let $\alpha_{n,\theta}$ be the continuity modulus of 
  $$
  B_{n,\theta}(u) = \int_{\mathbb{R}^p} \mathbf{1}(\theta'z \leq u) \d B_n(T(z)).
  $$
  With $\alpha_n$ being the continuity modulus of $B_n$, it follows that $\alpha_{n,\theta}(\epsilon) \leq \alpha_n(L_T\epsilon),$ where $L_T$ is the Lipschitz constant of $T$ from (F). Using standard arguments, one can show that
  $$
  \sup_{\lVert s - t\rVert \leq \epsilon p^{1/2}} \left[E\left\{(B_n(s) - B_n(t))^2\right\}\right]^{1/2} \leq q(\epsilon ) = \begin{cases} \sqrt{p\epsilon(1 - \epsilon)}, & 0 < \epsilon < 1/2 \\ \sqrt{p}/2, & \epsilon \geq 1/2. \end{cases}
  $$
  Then Lemma 2 of \cite{silverman1976gaussian} implies the existence of a random variable $A > 0$ with $E(A) < \infty$ such that
  \begin{equation}
  \label{eq:alphaBound}
  \alpha_n(\epsilon) \leq 16(p\log A)^{1/2}q(\epsilon) + 16p\sqrt{2}\int_0^\epsilon \left\{\log(1/r)\right\}^{1/2}\d q(r).
  \end{equation}
  Applying this bound to the integrand in \eqref{eq:rhoBound1}, we find that
  $$
  \sup_{x,\theta} |\rho_{n,\theta}(x)| \leq h\inv\int_{\mathbb{R}} \alpha_n(L_t|s|h) |\d \kappa(s)|.
  $$
  Hence, applying the exact arguments in Proposition~4 of \cite{silverman1978weak} under (K), we can conclude line 2 of \eqref{eq:BB_approx}.
  
\end{proof}

\begin{lemma}
\label{lma:unifKDE}
Under assumptions (K) and (F), 
$$
\sup_{\omega,x,\theta}\left|\frac{1}{n} \sum_{i = 1}^n\{ \hat{r}_h(X_i, x, \theta) - r_h(X_i, x, \theta)\}d^2(Y_i,\omega)\right| = O_P(b_n),
$$
where
$$
b_n = \max\left\{\left(\frac{-\log h}{nh}\right)^{1/2}, \frac{[\log n]^{3/2}}{hn^{(p + 2)/[2(p+1)]}}\right\}.
$$
\end{lemma}

\begin{proof}
  By applying Lemma~\ref{lma:BB_approx} to the kernels $\kappa(u) = K(u)u^j,$ $j = 0,1,2,$ we can conclude that
  $\hat{\mu}_{j,\theta}(x) = \mu_{j,\theta}(x) + h^jO_P(b_n)$ uniformly in $x$ and $\theta.$  Using the fact that $\mu_{j,\theta}(x) = h^j f_{\theta'X}(\theta'x)(K_{1j} + o(1))$ by (K) and (F), simple calculations show that $\hat{\sigma}^{-2}_\theta(x) = \sigma^{-2}_\theta(x) + O_P(b_nh^{-2})$ uniformly in $x$ and $\theta.$  Consequently, by assumption (F),
  \begin{align*}
  \hat{r}_h(z, x, \theta) - r_h(z, x, \theta) &= O_P(b_nh^{-2})K_h\left[\theta'(z - x)\right]\left\{\hmu_{2,\theta}(x) - \hmu_{1,\theta}(x)\theta'(z-x)\right\} \\
  & \hspace{-0.2cm} + O(h^{-2})K_h\left[\theta'(z-x)\right]\left\{O_P(b_nh^2) - O_P(hb_n)\theta'(z-x)\right\},
  \end{align*}
  where all $O(\cdot)$ and $O_P(\cdot)$ terms are uniform in $\theta$ and $x.$
  
  Finally, applying Lemma~\ref{lma:BB_approx} to the kernel $\kappa(u) = K(u)|u|$, similar analysis shows that
  $$
  \hmu_{j,\theta}^+(x) = n\inv\sum_{i = 1}^n K_h(\theta'(X_i - x))|\theta'(X_i - x)|^j
  $$
  satisfies $\sup_{x,\theta} |\hmu_{j,\theta}^+(x)| =  O_P(h^j).$ Thus,
   \begin{align*}
      \sup_{\omega,x,\theta} &\left|\frac{1}{n}\sum_{i = 1}^n \left\{\hat{r}_h(X_i, x, \theta) - r_h(X_i, x, \theta)\right\}d^2(Y_i,\omega)\right| \\
      &\hspace{1cm} \leq \mathrm{diam}^2(\Omega)\sup_{x,\theta} \left[O_P(b_nh^{-2})\{\hmu_{0,\theta}^+(x)\hmu_{2,\theta}(x) - \hmu_{1,\theta}(x)\hmu_{1,\theta}^+(x)\} \right . \\
      &\hspace{3cm} + \left. O(h^{-2})\{\hmu_{0,\theta}^+(x)O_P(b_nh^2) + \hmu_{1,\theta}^+(x)O_P(b_nh)\} \right] \\
      &\hspace{1cm}= O_P(b_n).
  \end{align*}
  
\end{proof}

\begin{lemma}
\label{lma:BB_approx2}
Under assumptions (K) and (F), and with $\kappa$ satisfying the conditions in Lemma~\ref{lma:BB_approx}, for any fixed $\omega \in \Omega,$ set 
$$
\hat{\nu}_\theta^+(x) = n^{-1}\sum_{i = 1}^n \kappa_h(\theta'(X_i - x))d^2(Y_i, \omega), \quad \kappa_h(\cdot) = h\inv\kappa(\cdot/h).
$$
Then, on a sufficiently rich probability spaces, there exist processes $\rho_{n,\theta}^+(x)$ and $\varepsilon_{n,\theta}^+(u)$ such that, provided $h \rightarrow 0$ as $n \rightarrow \infty,$
\begin{equation}
\label{eq:BB_approx2}
\begin{split}
    \hat{\nu}_\theta^+(x) &= E(\hat{\nu}_\theta^+(x)) + n^{-1/2}\rho_{n,\theta}^+(x) + n^{-1/2}\varepsilon_{n,\theta}^+(x), \\
    \sup_{x,\theta} |\rho_{n,\theta}^+(x)| &= O_P\left\{(-h\inv \log h)^{1/2}\right\}, \, \text{and} \\
    \sup_{x,\theta} |\varepsilon_{n,\theta}^+(x)| &= O\left(\frac{(\log n)^{3/2}}{hn^{1/[2(p+2)]}}\right) \, \text{a.s.} 
\end{split}
\end{equation}
\end{lemma}

\begin{proof}
  The proof follows the same line as that of Lemma~\ref{lma:BB_approx}, with some adjustments to deal with the presence of the response variable.  For any fixed $\omega,$ write $R = d^2(Y,\omega)$ and $R_i = d^2(Y_i,\omega).$  Furthermore, let $F_{X,R}$ and $\hat{F}_{X,R}$ be the population and empirical cumulative distribution functions of $(X_i,R_i)$, and similarly define $F_{\theta'X,R}$ and $\hat{F}_{\theta'X,R}$ for the random pairs $(\theta'X_i, R_i).$  Applying the result of \cite{csorgHo1975new} to the vectors $(X_i,R_i) \in \mathbb{R}^{p+1},$ and letting $T^+$ be the Rosenblatt transformation of $(X,R),$ there is a $(p + 1)$-dimensional Brownian bridge $B_n^+$ such that, with $\hat{G}^+$ denoting the empirical cumulative distribution function of $T^+(X_i,R_i),$ 
  \begin{equation}
      \label{eq:csorgo2}
      \sup_{t \in [0,1]^{p+1}}|\sqrt{n}(\hat{G}^+(t) - t) - B_n^+(t)| = O\left(\frac{[\log n]^{3/2}}{n^{1/[2(p+2)]}}\right) a.s.
  \end{equation} 
  
  Continuing, set $Z_n^+(x,r) = n^{1/2}(\hat{F}_{X,R}(x,r) - F_{X,R}(x,r)) - B_n^+(T^+(x,r))$, so that, for $u \in \mc{U}_\theta,$
  \begin{equation*}
  \begin{split}
  \hat{F}_{\theta'X,R}(u,r) - F_{\theta'X,R}(u,r) &= \int_{\mathbb{R}^p} \mathbf{1}(\theta'z \leq u) \d_z \left[\hat{F}_{X,R}(z,r) - F_{X,R}(z,r)\right] \\
  &=n^{-1/2}\int_{\mathbb{R}^p} \mathbf{1}(\theta'z \leq u) \d_z B_n^+(T^+(z,r)) \\
  &\hspace{1.5cm} + n^{-1/2}\int_{\mathbb{R}^p} \mathbf{1}(\theta'z \leq u) \d_z Z_n^+(z,r) \\
  &= n^{-1/2}B_{n,\theta}^*(u,r) + n^{-1/2}Z_{n,\theta}^*(u,r).
  \end{split}
  \end{equation*}
  Hence,
  \begin{equation}
      \label{eq:nu_exp2}
  \begin{split}
  \hat{\nu}_{\theta}^+(x) &= E(\hat{\nu}_\theta^+(x)) + \int_{\mathbb{R}^2} \kappa_h(u - \theta'x)r \d[\hat{F}_{\theta'X,R}(u,r) - F_{\theta'X,R}(u,r)] \\
  &= E(\hat{\nu}_\theta^+(x)) + n^{-1/2}\rho_{n,\theta}^+(x) + n^{-1/2}\varepsilon_{n,\theta}^+(x),
  \end{split}
  \end{equation}
  where 
  \begin{align*}
  \rho_{n,\theta}^+(x) &=  \int_{\mathbb{R}}\kappa_h(u - \theta'x)\d B_{n,\theta}^+(u), \quad B_{n,\theta}^+(u) = \int_{\mathbb{R}} r \d_r B_{n,\theta}^*(u,r), \\  \varepsilon_{n,\theta}^+(x) &= \int_{\mathbb{R}} \kappa_h(u - \theta'x)\d Z_{n,\theta}^+(u), \quad Z_{n,\theta}^+(u) = \int_{\mathbb{R}} r\d_r Z_{n,\theta}^*(u,r).
  \end{align*}
  From assumption (K) and \eqref{eq:csorgo2}, the third line of \eqref{eq:BB_approx2} is established since
  \begin{equation}
      \label{eq:varepBound}
      \sup_{\theta,x} |\varepsilon_{n,\theta}^+(x)| = O\left(\frac{[\log n]^{3/2}}{hn^{1/[2(p+2)]}}\right) \quad \textrm{a.s.}
  \end{equation}

  Next, consider the continuity modulus $\alpha_n^+$ of the Gaussian process 
  $$
  \tilde{B}_n^+(z) = \int_\mathbb{R} r \d_r B_n^+(T^+(z,r))
  $$ 
  on $\mathbb{R}^p$.  Define $\gamma_j(z) = \int_\mathbb{R} r^j \d_r F_{X,R}(z,r)$, $j = 1,2$, and, for two points $z,z',$ let $\underline{z}$ denote their element-wise minimum.  Then
  $$
  \mathrm{Cov}(\tilde{B}_n^+(z), \tilde{B}_n^+(z')) = \gamma_2(\underline{z}) - \gamma_1(z)\gamma_1(z').
  $$
  Hence,
  \begin{equation}
  \label{eq:meanSqCont}
  \begin{split}
  E\left\{(\tilde{B}_n^+(z) - \tilde{B}_n^+(z'))^2\right\} &= \gamma_2(z) + \gamma_2(z') - 2\gamma_2(\underline{z}) - [\gamma_1(z) - \gamma_1(z')]^2 \\
  &\leq \gamma_2(z) + \gamma_2(z') - 2\gamma_2(\underline{z}) \\
  &= \int_{\mathbb{R}} r^2 \d_r[F_{X,R}(z,r) - F_{X,R}(\underline{z},r)] \\
  &\hspace{0.5cm} + \int_{\mathbb{R}} r^2 \d_r[F_{X,R}(z',r) - F_{X,R}(\underline{z},r)]
  \end{split}
  \end{equation}
  Under assumption (F), since $\max\{\norm{z - \underline{z}},\norm{z' - \underline{z}}\} \leq \norm{z - z'},$ the above is bounded above by $M\norm{z - z'}$ for $M = \sup_{z} \int_{\mathbb{R}} r^2f_{X,R}(z,r)\d r.$
  Hence, we apply Lemma~2 of \cite{silverman1976gaussian} to conclude that there exists a random variable $A^+ > 0$ with finite expectation such that, with $q^+(\epsilon) = \epsilon^{1/2},$
  \begin{equation}
      \label{eq:alphaBound2}
      \alpha_n^+(\epsilon) \leq 16(2Mp^{1/2}\log A^+)^{1/2}q^+(\epsilon) + 32p^{3/4}\sqrt{M}\int_0^\epsilon \{\log(1/s)\}^{1/2} \d q^+(s).
  \end{equation}
  Letting $\alpha_{n,\theta}^+$ denote the continuity modulus of $B_{n,\theta}^+(u)$, the above arguments demonstrate that
  $$
  \alpha_{n,\theta}^+(\epsilon) \leq \alpha_n^+(L_T^+\epsilon),
  $$
  where $L_T^+$ is the Lipschitz constant of $T^+.$  
  
  Finally, using integration by parts, the change of variable $s = (u - \theta'x)/h,$ and assumption (K), for large enough $n$ we will have
  \begin{equation}
      \label{eq:rhobound2}
      \begin{split}
      |\rho_{n,\theta}^+(x)| &= \left|\int_{\mathbb{R}} \kappa_h(u - \theta'x)\d B_{n,\theta}^+(u)\right| \\
      &=h\inv \int_\mathbb{R}|B_{n,\theta}^+(\theta'x + hs) - B_{n,\theta}^+(\theta'x)||\d \kappa(s)| \\
      &\leq h\inv\int_{\mathbb{R}} \alpha_n^+(L_T^+hs) |\d \kappa(s)|.
      \end{split}
  \end{equation}
  Applying the integral arguments of Proposition~3 of \cite{mack1982weak} under (K), we can conclude the second line of \eqref{eq:BB_approx2}.
\end{proof}

\begin{lemma}
\label{lma:unifKernReg}
Under assumptions (K) and (F), for any fixed $\omega \in \Omega,$
$$
\sup_{x,\theta}\left|\frac{1}{n}\sum_{i = 1}^n r_h(X_i,x,\theta)d^2(Y_i,\omega) - \tilde{\Lambda}_\oplus(\omega,\theta'x,\theta)\right| = O_P\left(b_n'\right),
$$
where $$b_n' = \max\left[\left\{\frac{-\log h}{nh}\right\}^{1/2}, \frac{[\log n]^{3/2}}{hn^{(p + 3)/[2(p + 2)]}}\right].$$
\end{lemma}

\begin{proof}
  Define $\hat{\nu}_{0,\theta}^+(x)$ and $\hat{\nu}_{1,\theta}^+(x)$ as in the statement of Lemma~\ref{lma:BB_approx2} for the kernel choices $K(u)$ and $K(u)u,$ respectively.  Then write
  \begin{equation}
  \label{eq:rhDec}
  \begin{split}
  &\frac{1}{n}\sum_{i = 1}^n r_h(X_i,x,\theta)d^2(Y_i,\omega) - \tilde{\Lambda}_\oplus(\omega,\theta'x,\theta) \\
  &\hspace{1cm}= \frac{\mu_{2,\theta}(x)}{\sigma_\theta^2(x)}\left\{\hat{\nu}_{0,\theta}(x) - E[\hat{\nu}_{0,\theta}(x)]\right\} + \frac{h\mu_{1,\theta}(x)}{\sigma_\theta^2(x)}\left\{\hat{\nu}_{1,\theta}(x) - E[\hat{\nu}_{1,\theta}(x)]\right\}.
  \end{split}
  \end{equation}
  Under assumption (K) and (F), both $\mu_{2,\theta}(x)\sigma_{\theta}^{-2}(x)$ and $h\mu_{1,\theta}(x)\sigma_\theta^{-2}(x)$ are uniformly bounded in $x$ and $\theta$ as $h \rightarrow 0.$  Applying Lemma~\ref{lma:BB_approx2} to $\hat{\nu}_{j,\theta}(x),$ $j = 0,1,$ completes the proof.
\end{proof}

\subsection{Proofs of main results}
\label{ss:main_proofs}

\begin{proof}[Proof of Theorem~\ref{thm:unif_cons}] 
Begin by expanding
  \begin{equation}
    \label{eq:sup_crit}
    \begin{split}
    \sup_{\omega,x,\theta}& |\hLp(\omega,\theta'x,\theta) - \Lp(\omega,\theta'x,\theta)| \\
    &\leq \sup_{\omega,x,\theta}\left|\tilde{\Lambda}_\oplus(\omega,\theta'x,\theta) - \Lambda_\oplus(\omega,\theta'x,\theta)\right| \\
    &\hspace{0.5cm} + \sup_{\omega,x,\theta} \left|\frac{1}{n}\sum_{i = 1}^n \{\hat{r}(X_i, x,\theta) - r_h(X_i, x, \theta)\}d^2(Y_i, \omega)\right| \\
    &\hspace{0.5cm} + \sup_{\omega,x,\theta} \left|\frac{1}{n}\sum_{i = 1}^n r_h(X_i,x,\theta)d^2(Y_i,\omega) - \tilde{\Lambda}_\oplus(\omega,\theta'x,\theta)\right|.
    \end{split}
  \end{equation}
The first two terms on the right-hand side are $o_P(1)$ by Lemmas~\ref{lma:bias} and \ref{lma:unifKDE}, respectively.  For the third term, we have
$$
\psi_n(\omega) = \sup_{x,\theta} \left|\frac{1}{n}\sum_{i = 1}^n r_h(X_i,x,\theta)d^2(Y_i,\omega) - \tilde{\Lambda}_\oplus(\omega,\theta'x,\theta)\right| = o_P(1)
$$
for any fixed $\omega$ by Lemma~\ref{lma:unifKernReg} and the conditions on the bandwidth.  It is straightforward to show, using Lemma~\ref{lma:unifKDE} and the boundedness of $\Omega$ that, for any $\omega_1,\omega_2 \in \Omega,$
\[
\begin{split}
|\psi_n(\omega_1) - \psi_n(\omega_2)| &\leq O(d(\omega_1,\omega_2))\sup_{\theta,x}\left\{\frac{1}{n}\son |r_h(X_i,x,\theta)| + E[r_h(X,x,\theta)]\right\} \\
& = d(\omega_1,\omega_2)\left[O_P(b_n) + O(1)\right] \\
& = O_P(d(\omega_1,\omega_2)),
\end{split}
\]
where $b_n$ is the rate given in the statement of Lemma~\ref{lma:unifKDE}.  Hence, applying Theorem 1.5.4 of \cite{well:96}, we see that $\sup_{\omega \in \Omega} |\psi_n(\omega)| = o_p(1),$ so that
the first expression in \eqref{eq:sup_crit} converges to zero in probability.

Having established this, for any $\epsilon > 0,$ let $\eta > 0$ be the constant in (M).  Let $x,\theta$ be fixed, and $\omega$ be a point satisfying \mbox{$d(\omega, \gp(\theta'x, \theta)) > \epsilon.$}  If 
$$
\sup_{\theta, x, \omega}|\hLp(\omega,\theta'x,\theta) - \Lp(\omega,\theta'x,\theta)| < \eta/2,
$$ 
then
\[
\begin{split}
\hLp(\gp(\theta'x,\theta),\theta'x,\theta) &< \Lp(\gp(\theta'x,\theta),\theta'x,\theta) + \frac{\eta}{2} \\
&<\Lp(\omega,\theta'x,\theta) - \eta + \frac{\eta}{2} \\
&<\hLp(\omega,\theta'x,\theta) + \frac{\eta}{2} - \frac{\eta}{2} \\
&=\hLp(\omega,\theta'x,\theta).
\end{split}
\]
Therefore, such $\omega$ cannot be a minimizer of $\hLp(\cdot,\theta'x,\theta),$ whence 
$$
d(\hgp(\theta'x,\theta),\gp(\theta'x,\theta))\leq \epsilon.
$$  Since this argument holds simultaneously for all $\theta$ and all $x,$ the result holds.

\end{proof}

\begin{proof}[Proof of Corollary~\ref{cor:theta_cons}]
  Let $V(\theta_0)$ be any neighborhood of $\theta_0$ in $\Theta_p.$  then
  \[
  \begin{split}
    P(\hat{\theta} \in V(\theta_0)) \geq P(W_n(\hat{\theta}) &\geq W_n(\theta)) - P(\inf_{\theta \notin V(\theta_0)} W_n(\theta) \leq W_n(\theta_0)) \\
    &= 1 - P(\inf_{\theta \notin V(\theta_0)} W_n(\theta) \leq W_n(\theta_0))
  \end{split}
  \]
  We will show this last probability tends to zero.  Writing $$\tilde{W}_n(\theta) = n\inv \son d^2(Y_i, \gp(\theta'X_i,\theta)),$$
  it follows that
  \[
  \begin{split}
      &P\left(\inf_{\theta \notin V(\theta_0)} W_n(\theta) \leq W_n(\theta_0)\right) \\ &\hspace{0.5cm}\leq P\left(\inf_{\theta \notin V(\theta_0)} [ W_n(\theta) - \tilde{W}_n(\theta)] + \inf_{\theta \notin V(\theta_0)} [\tilde{W}_n(\theta) - W(\theta)] \right. \\
      &\hspace{2cm} \left.+ \inf_{\theta \notin V(\theta_0)} W(\theta) \leq  W_n(\theta_0)\right) \\
      &\hspace{0.5cm}\leq P\left(\sup_{\theta} |W_n(\theta) - \tilde{W}_n(\theta)| + \sup_{\theta}|\tilde{W}_n(\theta) - W(\theta)| \right. \\
      &\hspace{2cm} \left. + |W_n(\theta_0) - W(\theta_0)| > \inf_{\theta \notin V(\theta_0)} W(\theta) - W(\theta_0)\right).
  \end{split}
  \]
  As $\inf_{\theta \notin V(\theta_0)} W(\theta) - W(\theta_0) > 0$ due to the identifiability condition, we show that each of the terms on the left hand side of the last probability statement converges to zero in probability.  Using the uniform law of large numbers, which is applicable here since (M) implies continuity of $\gp$ in both arguments, $\sup_{\theta} |\tilde{W}_n(\theta) - W(\theta)| = o_p(1)$.  By boundedness of $\Omega$ and the triangle inequality,
  $$
  \sup_\theta |W_n(\theta) - \tilde{W}_n(\theta)| \leq \frac{2\mathrm{diam}(\Omega)}{n}\sup_\theta \son d(\gp(\theta'X_i,\theta), \hgp(\theta'X_i,\theta))
  $$
  is $o_P(1)$ by Theorem~\ref{thm:unif_cons}.  The above results clearly imply that $|W_n(\theta_0) - W(\theta_0)| = o_p(1)$.
\end{proof}
\begin{proof}[Proof of Corollary~\ref{cor:reg_cons}]
By direct application of the triangle inequality, whenever $\norm{\hat{\theta} - \theta_0} < \delta,$
\[
  \begin{split}
  \sup_{x} d(\hmp(x), \mp(x)) &= \sup_x d(\tgp(\hat{\theta}'x, \hat{\theta}), g_\oplus(\theta_0'x, \theta_0) \\
  & \leq \sup_{\norm{\theta - \theta_0} < \delta} \sup_x d(\tgp(\theta'x, \theta), g_\oplus(\theta'x, \theta)) \\
  &\hspace{1.5cm} +\,\, \sup_x d(g_\oplus(\hat{\theta}'x, \hat{\theta}), g_\oplus(\theta_0'x, \theta_0)).
  \end{split}
  \]
  By assumption, the first term converges to zero in probability.  Since $\hat{\theta}$ has been proven to be consistent and assumption (M) implies continuity of $g_\oplus$, the second term also converges weakly to zero.
\end{proof}

\bibliographystyle{plain}

\end{document}